\documentclass[conference]{IEEEtran}
\IEEEoverridecommandlockouts
\usepackage{times}
\usepackage{multirow}
\usepackage{subfigure}
\usepackage[linesnumbered,ruled, vlined]{algorithm2e}
\usepackage{amsmath,amsthm,amsfonts, bm}
\usepackage{array}
\usepackage{booktabs}
\usepackage{graphicx}
\usepackage{epsf}
\usepackage{url}
\usepackage{xcolor}
\usepackage{balance}
\usepackage{marginnote}
\usepackage{hyperref}

\everymath{\displaystyle}

\newcolumntype{L}[1]{>{\raggedright\let\newline\\\arraybackslash\hspace{0pt}}m{#1}}
\newcolumntype{C}[1]{>{\centering\let\newline\\\arraybackslash\hspace{0pt}}m{#1}}
\newcolumntype{R}[1]{>{\raggedleft\let\newline\\\arraybackslash\hspace{0pt}}m{#1}}

\newtheorem{definition}{Definition}
\newtheorem{theorem}{Theorem}
\newtheorem{lemma}{Lemma}
\newtheorem{example}{Example}

\newcommand{\CHENG}{}
\newcommand{\KAIXIN}{}
\newcommand{\chengfinal}{}

\newcommand{\KXREVIEW}{}
\newcommand{\CHENGREVIEW}{}

\AtBeginDocument{%
  \providecommand\BibTeX{{%
    \normalfont B\kern-0.5em{\scshape i\kern-0.25em b}\kern-0.8em\TeX}}}

\begin{document}


\title{Reinforcement Learning Enhanced Weighted Sampling for Accurate Subgraph Counting on Fully Dynamic Graph Streams}


\author{
\IEEEauthorblockN{
Kaixin Wang\IEEEauthorrefmark{1}, 
Cheng Long\IEEEauthorrefmark{1}\IEEEauthorrefmark{4}\IEEEcompsocitemizethanks{\IEEEcompsocthanksitem\IEEEauthorrefmark{4}Corresponding author.},
Da Yan\IEEEauthorrefmark{2}, 
Jie Zhang\IEEEauthorrefmark{1}, 
H. V. Jagadish\IEEEauthorrefmark{3}
}
\IEEEauthorblockA{
\IEEEauthorrefmark{1}School of Computer Science and Engineering, Nanyang Technological University, Singapore\\
\IEEEauthorrefmark{2}Department of Computer Science, University of Alabama at Birmingham, United States\\
\IEEEauthorrefmark{3}Computer Science and Engineering, University of Michigan, United States\\
\{kaixin.wang, c.long, zhangj\}@ntu.edu.sg, 
yanda@uab.edu,
jag@umich.edu
}
}


\maketitle


\begin{abstract}
    As the popularity of graph data increases, there is a growing need to count the occurrences of subgraph patterns of interest, for a variety of applications. Many graphs are massive in scale and also fully dynamic (with insertions and deletions of edges), rendering exact computation of these counts to be infeasible. Common practice is, instead, to use a small set of edges as a sample to estimate the counts. Existing sampling algorithms for fully dynamic graphs sample the edges with uniform probability. In this paper, we show that we can do much better if we sample edges based on their individual properties. Specifically, we propose a weighted sampling algorithm called \texttt{WSD} for estimating the subgraph count in a fully dynamic graph stream, which samples the edges based on their weights that indicate their importance and reflect their properties. We determine the weights of edges in a data-driven fashion, using a novel method based on reinforcement learning. We conduct extensive experiments to verify that our technique can produce estimates with smaller errors while often running faster compared with existing algorithms.
    
\end{abstract}



\section{Introduction}
\label{sec:intro}

Graphs have been widely used to represent the data structures in online social networks (e.g., Facebook and Twitter) and Internet applications (e.g., Youtube, blogs, and web pages), where vertices represent individuals or entities and edges represent interactions or connections among them. 
Counting a certain subgraph pattern (e.g., triangle) on these graphs can help reveal network structure information \cite{eckmann2002curvature}, detect anomalous behavior \cite{lim2015mascot, shin2020fast} and identify the interests of users \cite{leskovec2007dynamics, zhao2016link}.
{\KXREVIEW{
For example, in social network analysis, a triangle has been {\chengfinal proven} to be an evidence of the phenomena of homophily~\cite{mcpherson2001birds, aiello2012friendship} (i.e., people tend to make friends with who are similar to them) and transitivity~\cite{wasserman1994social, jamali2011transitivity} (i.e., people who share common friends become friends). Therefore, many concepts in social network, such as clustering coefficient~\cite{buriol2006counting} and transitivity ratio~\cite{jha2013space}, are based on the triangle count. 
%
Moreover, 
\cite{kang2011spectral, kang2012heigen} show that in web networks,
normal nodes usually have similar and mild ratios of the triangle counts to the degrees {\CHENG whereas} spammers or harmful accounts are often linked with fewer but remarkably well-connected nodes, resulting in extremely high ratios.
Therefore, {\CHENG they propose to detect} anomalous nodes and structures in email systems and phone call networks {\CHENG based on the triangle counts and degrees information, and the effectiveness of the proposed method has been verified via case studies~\cite{lim2015mascot}}. 
}}
Among the existing studies for estimating the count of a certain subgraph pattern, some focus on insertion-only graph streams~\cite{lim2015mascot, ahmed2017sampling, zhang2019t} and others on fully dynamic graph streams (which involve both insertions and deletions of edges)~\cite{stefani2017triest, shin2020fast, lee2020temporal, shin2017wrs, shin2018think}.
In this paper, we study the problem on fully dynamic graph streams since they are more general than insertion-only graph streams and in many of these applications, the graphs are in the form of fully-dynamic graph streams with edges inserted and deleted dynamically. 
For example, in a social network such as Facebook, connections and disconnections among users would happen as time goes by, which correspond to edge insertions and deletions, respectively. 

We observe that the existing studies for fully dynamic graph streams all sample the edges with uniform probabilities, e.g., each edge is treated equally for sampling~\cite{ shin2020fast, stefani2017triest, lee2020temporal, shin2017wrs, shin2018think}, which would likely result in sub-optimal samples.
%
To illustrate, suppose that we would like to estimate the number of triangles (3-cliques) in a social network like Twitter. There would be more triangles involving two celebrities if they subscribe each other, and thus the edge between them should be sampled with a higher probability than those among the generic public so as to achieve smaller estimation variance.

Motivated by the phenomenon that different edges have different importance and should be assigned with different probabilities for sampling, in this paper, we aim to develop a \emph{weighted} sampling algorithm for estimating the count of a certain subgraph structure on fully dynamic graph streams. 
In a weighted sampling algorithm, each edge would be assigned with a \emph{weight} indicating the importance of the edge such that edges with higher weights would be sampled with higher probabilities.
This immediately gives rise to two problems. First, how can we perform weighted sampling on a fully dynamic graph stream? Second, while we can all intuitively see that more important edges should be weighted higher, what exactly should these weights be? In this paper, we address both these problems in turn.

To address the first problem, we first review an existing weighted sampling framework called \texttt{GPS}~\cite{ahmed2017sampling} for insertion-only graph streams. 
{\CHENG{Specifically, \texttt{GPS} maintains a reservoir of a fixed size $M$. Whenever an edge is inserted, it computes a rank of the edge in a probabilistic manner such that {\chengfinal it would be more likely that} the rank {\chengfinal is} higher if the edge's weight is higher. It then samples the edge if its rank is among the top-$M$. However, \texttt{GPS} is not capable of dealing with fully-dynamic graph streams since an edge with its rank not among the top-$M$ may be sampled if some space of the reservoir has been released due to some edge deletions that happened earlier.}}
We then present an adaption of \texttt{GPS} called \texttt{GPS-A}, which works similarly as \texttt{GPS} for edge insertions and attaches to each edge to be deleted a binary tag without practically deleting it. \texttt{GPS-A} works for fully-dynamic graph streams with both insertions and deletions, but is lazy to clear up the space taken for storing edges that have been deleted in a reservoir, which results in low accuracy. 
Finally, we introduce a carefully designed weighted sampling framework called \emph{weighted sampling with deletions} (\texttt{WSD}) for fully-dynamic graph streams. \texttt{WSD} extends \texttt{GPS} in a smart and careful way such that it can handle both edge insertions and deletions and avoid the drawback of \texttt{GPS-A}, i.e., the edges that have been deleted would not be stored in a reservoir.
{\CHENG{One intuition behind \texttt{WSD} is that it maintains a rank threshold and samples an edge only if its rank exceeds the threshold. The threshold is updated properly such that any two edges with equal weights would be sampled with equal probabilities.}}
We then construct an estimator of the count of a certain subgraph pattern based on the sampled edges in the reservoir by \texttt{WSD} and prove the \emph{unbiasedness} of the estimator.
%

For the second problem, the conventional approach has been to use heuristics to set the weights. For example, in the weighted sampling method \texttt{GPS} on insertion-only graph streams, the strategy is to set the weight of an edge to be number of connections between the edge and those in the reservoir or the number of subgraph structures formed by the edge and those in the reservoir. Nevertheless, the heuristic-based methods are non-adaptive to the underlying dynamics.
To fully unleash the power of \texttt{WSD}, we further propose an adaptive way of setting the weights of edges via reinforcement learning (RL).
Specifically, we regard the sampling process for subgraph counting problem as a sequential decision making process, i.e., we make a decision on how to set the weight of each new edge. 
We then model the sequential decision process as a Markov Decision Process (MDP)~\cite{puterman2014markov} and use a policy gradient method to learn a policy for the MDP. 
We carefully design the MDP including states and rewards such that (1) the states capture both temporal and topological information of the edge dynamics and can be easily computed with the sampled edges, and (2) the objective of maximizing the rewards of the policy is consistent with the goal of minimizing the estimation error. 

In summary, the main contributions of this paper are as follows. 

\begin{itemize}
    \item We propose a new fixed-size, weight-sensitive, one-pass sampling method \texttt{WSD} to handle fully dynamic graph streams. \texttt{WSD} is the first weighted sampling method for fully dynamic graph steams. 
    Based on \texttt{WSD}, we construct an estimator of the count of a given subgraph structure and prove the unbiasedness of the estimator. (Section~\ref{sec:sampling})
    \item We develop a reinforcement learning-based method for setting the weights of edges in \texttt{WSD} in a data-driven fashion, which is superior over heuristic-based methods. (Section~\ref{sec:rl})
    \item We conduct extensive experiments on several real graphs, including community networks, citation graphs, social networks and web graphs to verify that the new sampling framework \texttt{WSD} and the RL-based weighting method work better than the state-of-the-art methods. For example, \texttt{WSD} improves the effectiveness by 25\% - 47\% and often runs faster than the state-of-the-art method \texttt{WRS}~\cite{lee2020temporal, shin2017wrs}. (Section~\ref{sec:exp})
\end{itemize}

For the rest of the paper, we present the problem definition in Section~\ref{sec:preliminary}, review the related work in Section~\ref{sec:related}, and conclude the paper in Section~\ref{sec:conclusion}.


\section{Preliminaries and Problem Definition}
\label{sec:preliminary}

We first review the problem of counting subgraphs in a fully dynamic graph stream~\cite{stefani2017triest}. 
%
Consider a dynamic graph $G$, which evolves over time with the edges being inserted and deleted dynamically. 
Formally, the graph $G$ can be modelled as an edge stream $S = \{s^{(1)}, s^{(2)}, \cdots\}$, where $s^{(t)} = (op, e_t)$ represents an event of inserting (indicated by $op = +$) or deleting (indicated by $op = -$) the edge $e_t$.
We denote by $S^{(t)}$ the sequence of the first $t$ edge events, i.e., $S^{(t)}=\{s^{(1)}, s^{(2)}, \cdots, s^{(t)}\}\subset S$.
Let $G^{(t)}=(V^{(t)}, E^{(t)})$ be the induced graph from $S^{(t)}$. 
We assume that all edge events are feasible, i.e., if $e\in E^{(t)}$ (resp. $e\notin E^{(t)}$), $s^{(t+1)}$ cannot be $(+, e)$ (resp. $(-,e)$). 

Consider a certain subgraph pattern $H$ (e.g., a triangle). We use $|H|$ to denote the number of edges in $H$. 
Given a graph $G^{(t)}=(V^{(t)},E^{(t)})$, we use the notation $\mathcal{J}^{(t)}$ to denote the set of all subgraphs which are isomorphic to $H$ in $G^{(t)}$. 
Each subgraph $J\in \mathcal{J}$ can be uniquely identified by a set of ordered edges $J=\{e_{i_1}, e_{i_2}, \cdots, e_{i_{|H|}}\}$ with $i_1 < i_2 < \cdots < i_{|H|}$ being the arrival order of these edges. 
%

\begin{definition}[Subgraph Counting in Fully Dynamic Graph Streams \cite{stefani2017triest, lee2020temporal}]
Given a fully dynamic graph stream $S$, the problem is to estimate $|\mathcal{J}^{(t)}|$ for any $t$ accurately with the following constraints.
\begin{itemize}
    \item \textbf{No Knowledge}. We have no knowledge about the stream (e.g., the size of stream, the number of vertices and edges, etc.) in advance. 
    \item \textbf{{\chengfinal Limited} Memory}. We can store at most $M$ edges in a reservoir, where $M$ is a predefined parameter and independent to the size of the stream.
    \item \textbf{Single Pass}. Edge insertions and deletions are processed one by one in their arrival order. Edges cannot be accessed again once they are discarded. 
\end{itemize}
\label{def:problem}
\end{definition}

{\CHENGREVIEW 
These three constraints are inherited from existing studies~\cite{stefani2017triest, lee2020temporal}, on which we elaborate as follows. First, the ``no knowledge'' constraint naturally holds in many social network analytics applications, e.g., we normally cannot know how many connections would be formed in a social network in the future. Second, the ``limited memory'' constraint often holds when (1) we aim to achieve real-time responses (which is only possible when a limited number of edges are stored and used for estimating the properties of a graph) and/or (2) the graph is huge (e.g., it is web-scale and cannot fit in main memory in many cases). Third, the ``single pass'' constraint holds when we have a very high efficiency requirement since scanning a stream multiple times would be costly.}

\section{Weighted Sampling Frameworks and Subgraph Count Estimators}
\label{sec:sampling}

In this section, we first review an existing weighted sampling framework called \texttt{GPS}~\cite{ahmed2017sampling}, which only works for insertion-only graph streams, in Section~\ref{subsec:gps}. We then present a straightforward adaption of \texttt{GPS} called \texttt{GPS-A}, which works similarly as \texttt{GPS} for edge insertions and attaches to each edge to be deleted a binary tag without practically deleting it, in Section~\ref{subsec:GPS-A}. \texttt{GPS-A} works for fully-dynamic graph streams with both insertions and deletions, but is lazy to clear up the space taken for storing edges that have been deleted in a reservoir, which results in low accuracy. Finally, we introduce a carefully designed weighted sampling framework called \emph{weighted sampling with deletions} (\texttt{WSD}) for fully-dynamic graph streams in Section~\ref{subsec:wsd}. \texttt{WSD} extends \texttt{GPS} in a smart and careful way such that it can handle both edge insertions and deletions and avoid the drawback of \texttt{GPS-A}, i.e., the edges that have been deleted would not be stored in a reservoir.

%
\subsection{\texttt{GPS} Framework}
\label{subsec:gps}

\noindent\textbf{Sampling Process.}
\texttt{GPS} framework follows the \emph{priority sampling scheme}~\cite{duffield2007priority} and aims to sample a fixed-size reservoir of edges via a single pass such that edges that are deemed more important are sampled with higher probabilities. 
Specifically, when handling an event $(+, e_t)$ at time step $t$, it involves three steps. \underline{First}, it assigns the edge $e_t$ an appropriate \emph{weight} denoted by $w(e_t)$, based on the current reservoir $\mathcal{R}^{(t-1)}$. 
For example, it sets the weight to be the number of subgraph structures (e.g., triangles) that would be newly formed by $e_t$ as an indicator of the importance of $e_t$ - the larger the number is, the higher the weight is~\cite{ahmed2017sampling}.
%
We denote by $W(e, \mathcal{R})$ the function that computes the weight of an edge $e$ based on a reservoir $\mathcal{R}$. 
\underline{Second}, it samples a value $u$ from $(0, 1]$ uniformly and computes a \emph{rank} for the edge $e_t$, denoted by $r(e_t)$, based on $w(e_t)$ and $u$. 
For an edge with a higher weight, its rank would be likely higher - the randomness here is due to the sampling process of the value $u$.
%
%
We denote by $r=f(w)$ the function that computes the rank of an edge with the weight $w$. 
For example, $f(w)=w/u$ is used as the rank function in \cite{ahmed2017sampling}, where $u$ is a random value uniformly sampled in range $(0, 1]$.
\underline{Third}, it includes the edge $e_t$ in the reservoir if either (1) the reservoir is not full or (2) the rank of $e_t$ is larger than the smallest rank of an edge in the reservoir (in this case the edge with the smallest rank in the reservoir would be dropped due to the capacity limit of the reservoir); otherwise, it does not include the edge $e_t$ in the reservoir. For this step, it uses a minimum priority queue with the size equal to $M$ and the keys to be the ranks of edges in the queue.

With \texttt{GPS}, at the end of time step $t$, for each edge $e$ that has been inserted, the probability that edge $e$ is included in the reservoir, i.e., $\mathcal{R}^{(t)}$, is equal to the probability that $e$'s rank, i.e., $r(e)$, is larger than the $(M+1)^{th}$ largest rank among those of edges that have been inserted {\chengfinal (including $e$)}, which we denote by $r_{M+1}$. For $t \le M$, we define $r_{M+1}$ to be equal to 0. Specifically, we have the following equation.
\begin{equation}
\label{eq:gps-prob}
    \mathbb{P}[e \in \mathcal{R}^{(t)}] = 
    \mathbb{P}[r( e) > r_{M+1}]
\end{equation}
{\KAIXIN Note that this probability depends on the rank function. For example in \cite{ahmed2017sampling}, $r = f(w) = w/u$, then $\mathbb{P}[r > r_{M+1}]=\min\{1, w/r_{M+1}\}$. }
The above equation has been {\chengfinal proven} in~\cite{ahmed2017sampling}. Here, we provide some intuitive explanations. 
For $t\le M$, each edge that has been inserted would be included in the reservoir for sure, and thus the equation holds. For $t > M$, the reservoir stores those edges with the {\chengfinal $M^{th}$} largest ranks, and thus the probability that an edge to be included in the reservoir should be equal to the probability that its rank is larger than the $(M+1)^{th}$ largest rank of the edges that have been inserted.


\smallskip
\noindent\textbf{Estimator and Analysis.}
{\KXREVIEW{
It has been {\chengfinal proven} in \cite{ahmed2017sampling} that the probability that a set of edges $E=\{e_1, \cdots, e_{|E|}\}$ $(|E| \le M)$ is included in the reservoir $\mathcal{R}$ at the end of time $t$ is as follows. 
\begin{equation}
\label{eq:gps-edges}
    \mathbb{P}[E\subset \mathcal{R}] = \prod_{e\in E} \mathbb{P}[e\in \mathcal{R}] = \prod_{e\in E} \mathbb{P}[r(e) > r_{M+1}]
\end{equation}
where $r_{M+1}$ is observed at $t$. 
}}
Let $J=\{e_{i_1}, e_{i_2}, \cdots, e_{i_{|H|}}\}$ be a subgraph pattern formed at time $t_a(J)$. Note that $t_a(J)$ is the time at which the last edge of $J$ appears, i.e., $t_a(J)=i_{|H|}$. We define a random variable $X_{\texttt{GPS}}^{J}$ for $J$ as follows.
\begin{equation}
    X_{\texttt{GPS}}^{J}= \displaystyle\prod \limits_{e \in J \setminus e_{i_{|H|}}} \frac{\mathbb{I}(e\in\mathcal{R})}{\mathbb{P}[r( e) > r_{M+1}]}
\end{equation}
where $\mathbb{I}(\cdot)$ is an indicator function, 
and $r_{M+1}$ and $\mathcal{R}$ are the $(M+1)^{th}$ largest rank and the reservoir observed just after time $t_a(J)-1$. 
%
We define an estimator of the count of subgraph structures at
any time $t$, denoted by $c_{\texttt{GPS}}^{(t)}$, by the following equation.
\begin{equation}
    c_{\texttt{GPS}}^{(t)} = \sum_{J \in \mathcal{A}^{(t)}} X_{\texttt{GPS}}^{J}
\end{equation}
where $\mathcal{A}^{(t)}$ is set of subgraphs which are isomorphic to $H$ and have been added to the graph $G^{(t)}$ by time $t$. 
In~\cite{ahmed2017sampling}, it has been shown that the above estimator $c_{\texttt{GPS}}^{(t)}$ is \emph{unbiased}.
\begin{theorem}[Unbiasedness of the \texttt{GPS} estimator \cite{ahmed2017sampling}]
Given the graph stream $S$, which only consists of edge insertion events, and $M\geq |H|$, 
$\forall t$,
we have
\begin{equation}
    \mathbb{E}[c_{\texttt{GPS}}^{(t)}] = |\mathcal{A}^{(t)}| = |\mathcal{J}^{(t)}|.
\end{equation}
\label{theo:gps}
\end{theorem}





\vspace{-2mm}
\noindent\textbf{Inapplicability of \texttt{GPS} for Fully Dynamic Graph Streams.}
Unfortunately, \texttt{GPS} cannot be applied to fully dynamic graph streams, which involve both edge insertions and edge deletions. The reason is as follows. 
The correctness of \texttt{GPS} relies on the fact it guarantees that the edges with equal weights would be included in the reservoir with the equal probabilities (as shown in Eq.~(\ref{eq:gps-prob})). However, this would no longer be guaranteed when \texttt{GPS} is applied to fully dynamic graph streams directly. To illustrate, consider a scenario below. 
%
\begin{example}
Consider an edge stream where (1) all edges are assigned with equal weights, (2) at time $t$ ($t > M+1$) the first event of an edge deletion, i.e., $(-, e_{t})$, happens, and (3) at time $t' = t +1$, an event of edge insertion $(+, e_{t'})$ happens.
Let $p$ and $p'$ be the probability that an edge is included in the reservoir at time $t$ and $t'$, respectively. 
Obviously, we have $0<p' \le p < 1$ since the probability that an edge is included in the reservoir cannot be increasing.
Note that probabilities $p$ and $p'$ are shared by all edges since they have equal weights. 
{\KAIXIN Let $p''$ be the probability that the edge $e_{t'}$ is inserted into the reservoir when the reservoir is full at time $t'$. }
{\chengfinal Note that $p''$ is different from $p'$.}
%
%
We then deduce that for edge $e_{t'}$, the probability that it is included in the reservoir should be equal to $p\cdot 1 + (1-p)\cdot p''$. For the term $p\cdot 1$, it corresponds to the case that $e_{t}$ has been included in the reservoir and then deleted from the reservoir at time $t$ (and thus the probability for this case is equal to $p$), and in this case, $e_{t'}$ would be included in the reservoir for sure (i.e., with the probability 1) since
(1) the reservoir would be not full when inserting $e_{t'}$, and (2) \texttt{GPS} would unconditionally include an edge when the reservoir is not full. 
For the term $(1-p)\cdot p''$, it corresponds to the case that $e_{t}$ has not been included in the reservoir at time $t$ (and thus the probability for this case is equal to $(1-p)$), and in this case, $e_{t'}$ would be included in the reservoir with probability $p''$ {\chengfinal (by definition of $p''$)}. 
In conclusion, the probability that $e_{t'}$ would be included in the reservoir at time $t'$ is larger than that for other edges since $p\cdot 1 + (1-p)\cdot p'' > p \ge p'$ (though all edges are assigned with equal weights), which implies that \texttt{GPS} would fail in this scenario.
\end{example}

\subsection{\texttt{GPS-A} Framework}
\label{subsec:GPS-A}


\noindent\textbf{Sampling Process.} 
%
In \texttt{GPS-A}, the sampling process is exactly the same as that of \texttt{GPS} except that when a deletion event happens on an edge in the reservoir, we only attach a ``DEL'' tag to the edge, but we do not remove the edge from the reservoir. Equivalently speaking, we \emph{ignore} the deletion operations during the sampling process first (by attaching the tags to edges only), and when constructing the estimator, we would \emph{neglect} those edges with the tags. In this way, the probabilities of including the edges with equal weights in the reservoir would be equal, in the same way as \texttt{GPS}. The drawback of \texttt{GPS-A} is that some space of the reservoir, which is taken by the edges with the ``DEL'' tags and can be used for including other edges otherwise, would be wasted.

\smallskip
\noindent\textbf{Estimator and Analysis.}
{\KXREVIEW{
Since \texttt{GPS-A} simply attaches a tag to {\CHENG each of} the deleted edges, Eq.~(\ref{eq:gps-edges}) also holds for {\CHENG \texttt{GPS-A}}.
}}
We denote by $\mathcal{R}_{tag}$ the set of edges with the ``DEL'' tags in the reservoir $\mathcal{R}$. Let $J=\{e_{i_1}, e_{i_2}, \cdots, e_{i_{|H|}}\}$ be a subgraph formed at time $t_a(J)$. Note that $t_a(J)$ is the time at which the last edge of $J$ appears, i.e., $t_a(J)=i_{|H|}$.  We define a random variable $X_{\texttt{GPS-A}}^{J}$ for $J$ as follows. 
\begin{equation}
    X_{\texttt{GPS-A}}^{J}= \displaystyle\prod \limits_{e \in J \setminus e_{i_{|H|}}} \frac{\mathbb{I}(e\in \mathcal{R}\setminus \mathcal{R}_{tag} )}{\mathbb{P}[r( e) > r_{M+1}]}
\end{equation}
where $\mathbb{I}(\cdot)$ is an indicator function, $\mathcal{R}$ and $r_{M+1}$ are the reservoir and the {\chengfinal $(M+1)^{th}$} largest rank of an edge among those that have appeared, as observed just after time $t_a(J)-1$. 
%
Assume that a subgraph $J$ is destroyed at time $t_d(J)$ when a deletion event happens on an edge $e_x$ of $J$. We define another random variable $Y_{\texttt{GPS-A}}^J$ for $J$ as follows. 
\begin{equation}
    Y_{\texttt{GPS-A}}^J= \displaystyle\prod \limits_{e \in J \setminus e_x} \frac{\mathbb{I}(e\in\mathcal{R} \setminus \mathcal{R}_{tag})}{\mathbb{P}[r( e) > r_{M+1}]}
    \label{equation:estimator-2}
\end{equation}
where 
$\mathcal{R}$ and $r_{M+1}$ are the reservoir and the {\chengfinal $(M+1)^{th}$} largest rank of an edge among those that have appeared, as observed just after time $t_d(J)-1$. 
Based on these two subgraph estimators, we define an estimator of the count of subgraph structures at any time $t$, denoted by $c_{\texttt{GPS-A}}^{(t)}$, as follows.  
\begin{equation}
    c_{\texttt{GPS-A}}^{(t)} = \sum_{J \in \mathcal{A}^{(t)}} X_{\texttt{GPS-A}}^J - \sum_{J \in \mathcal{D}^{(t)}} Y_{\texttt{GPS-A}}^J
\end{equation}
where $\mathcal{A}^{(t)}$ (resp. $\mathcal{D}^{(t)}$) is set of subgraphs which are isomorphic to $H$ and have been added to (resp. deleted from) the graph $G^{(t)}$ by time $t$.
\begin{theorem}[Unbiasedness of the subgraph count estimator of \texttt{GPS-A}] 
\label{theo:gps-a}
Given the graph stream $S$ and $M\geq |H|$, 
$\forall t$,
we have
\begin{equation}
    \mathbb{E}[c_{\texttt{GPS-A}}^{(t)}] = |\mathcal{A}^{(t)}| -  |\mathcal{D}^{(t)}| = |\mathcal{J}^{(t)}|.
\end{equation}
\end{theorem}
%

{\KXREVIEW{
The unbiasedness of $c_{\texttt{GPS-A}}^{(t)}$ can be verified based on Eq.~(\ref{eq:gps-edges}). Detailed proof can be found in Appendix~\ref{app:proof-gps-a}.
}}
{\KXREVIEW{
\begin{theorem}[Complexities of \texttt{GPS-A}]
\label{theo:gps-a-complexity}
The time and space complexity of \texttt{GPS-A} is $O((|A|+|D|) \cdot \log M \cdot \gamma(M))$ and $O(M)$, respectively, where $|A|$ (resp. $|D|$) {\CHENG is the number} of insertion (resp. deletion) events in the stream, $M$ is the maximum size of the reservoir, and $\gamma(M)$ is the time complexity of enumerating the subgraphs {\CHENG formed by} the sampled edges.
\end{theorem}

\begin{proof}
The time complexity is as follows. Consider an event of {\CHENG inserting} an edge $e=(u,v)$.
Finding the subgraphs formed by $e$ and some sampled edges would cost $\gamma(M)$. For example, for triangle counting, $\gamma(M)=O(|\mathcal{N}(u)|+|\mathcal{N}(v)|)=O(M)$ {\CHENG (which corresponds to the cost of computing the intersection between $\mathcal{N}(u)$ and $\mathcal{N}(v)$)}, where $\mathcal{N}(u)$ (resp. $\mathcal{N}(v)$) is the {\chengfinal set} of the neighbors of $u$ (resp. $v$) in the sampled graph. 
Then, after calculating $e$'s weight and rank, it takes $O(\log M)$ time to add the edge into a priority queue if it is included in the reservoir. Consider an event of edge deletion. Similarly, finding the subgraphs that are destroyed with this edge deletion would cost $\gamma(M)$. Then, we need to identify the position of the edge in the reservoir and assign it the tag if it is sampled, which costs $O(\log M)$. Thus, the total time cost is $O((|A|+|D|)\cdot \log M \cdot \gamma(M) )$. 
The space complexity is $O(M)$ since we can store at most $M$ edges in the reservoir. 
\end{proof}
}}

\smallskip
\noindent\textbf{Drawbacks of \texttt{GPS-A}.}
As mentioned before, \texttt{GPS-A} has an intrinsic drawback, i.e., those edges that have been included in the reservoir and then deleted would occupy some of storage of the reservoir without any benefits (for estimating the count of the subgraph structures). As the sampling process goes on, the reservoir would become smaller and smaller practically, which would result in low accuracy. Furthermore, it requires some extra space for storing the ``DEL'' tags.

\subsection{\texttt{WSD} Framework}
\label{subsec:wsd}

\begin{algorithm}[t]
\small
\DontPrintSemicolon
\SetNoFillComment
\caption{\texttt{WSD} Framework}
\label{alg:sampling}
\KwIn{An edge event stream $S=\{s^{(1)}, s^{(2)}, \cdots\}$.}
\KwOut{A reservoir $\mathcal{R}$ of sampled items. }
\SetKwFunction{insert}{insert}
\SetKwFunction{delete}{delete}
\SetKwProg{Fn}{function}{}{}
Let $\mathcal{R}$ be a min-priority queue with the maximum size $M$\;
$\mathcal{R} \gets \Phi$, $\tau_p \gets 0$, $\tau_q \gets 0$\;
\ForEach {$(op, e_t) \in S$} {
    \uIf {$op = +$} {
        \insert{$e_t$}\;
    }
    \Else {
        \delete{$e_t$}\;
    }
}

\Fn (\tcp*[f]{Case 1 \& Case 2}) {\insert{$e$}} {
    $w(e) \gets W(e, \mathcal{R})$\;
    $r(e) \gets f(w(e))$\;
    
    \uIf {$|\mathcal{R}| < M$} {    
        \If (\tcp*[f]{Case 1.1}) {$r(e) > \tau_p$} {
            $\mathcal{R} \gets \mathcal{R} \cup e$\;
        }
    }
    \Else  { 
        $e_m \gets$ the edge with the minimum rank in $\mathcal{R}$\;
        $\tau_p \gets r(e_m)$\; 
        \uIf (\tcp*[f]{Case 2.1}) {$r(e) > \tau_p$} {
            $\mathcal{R} \gets (\mathcal{R} \setminus e_m) \cup e$\;
            $\tau_q \gets \tau_p$\;
        }
        \ElseIf (\tcp*[f]{Case 2.2}) {$r(e) >\tau_q$}{
            $\tau_q \gets r(e)$\;
        }
    }
}

\Fn (\tcp*[f]{Case 3}) {\delete{$e$}} {         
    \If {$e \in \mathcal{R}$} {     
        $\mathcal{R} \gets \mathcal{R} \setminus e$\;
    }
}
\end{algorithm}

\noindent\textbf{Sampling Process.}
In \texttt{WSD}, we also use a min-priority queue with a fixed size of $M$ for storing the sampled edges. To avoid the drawbacks of \texttt{GPS-A}, when the deletion of an edge happens, we directly remove the edge from the reservoir if it exists in the reservoir. To ensure the correctness, we maintain two variables, namely $\tau_p$ and $\tau_q$, throughout the sampling process. 
\begin{itemize}
	\item \textbf{$\tau_p$.} We use it as a rank threshold such that for each insertion event $(+, e_t)$ at time $t$, we sample edge $e_t$ only if $e_t$'s rank is larger than $\tau_p$.
	
	\item \textbf{$\tau_q$.} It is a rank value for computing the probability with which an edge is sampled to the reservoir at the end of time $t$.
\end{itemize}
We maintain $\tau_q$ and $\tau_p$ appropriately such that at the end of time $t$, for each edge $e$ that has been inserted by time $t$ (inclusively) and not deleted yet, the probability $e$ is sampled in the reservoir is equal to the probability that $e$'s rank is larger than $\tau_q$, i.e., 

\begin{equation}
\label{eq:prob-new}
    \mathbb{P}[e\in\mathcal{R}^{(t)}] = 
    \mathbb{P}[r(e)>\tau_q]
\end{equation} 
{\KAIXIN Similarly, this probability depends on the rank function. If we adopt $r = f(w) = w/u$, then $\mathbb{P}[r > \tau_q]=\min\{1, w/\tau_q\}$.}
We will prove Eq.~(\ref{eq:prob-new}) formally later in Lemma~\ref{lemma:prob}.
{\CHENG{We note that $\tau_p$ (resp. $\tau_q$) is not always equal to the $M^{th}$ (resp. $(M+1)^{th}$) rank of those edges that have been inserted, which reflects one of the differences between \texttt{WSD} and \texttt{GPS-A}.}}

Specifically, we present the sampling process of \texttt{WSD} as follows. First, we initialize the queue to be empty and both variables $\tau_p$ and $\tau_q$ to be 0. Then, we present the sampling process in the following cases.
%
\begin{itemize}
	\item \textbf{Case 1: For an event $(+, e_t)$ with a \emph{non-full} reservoir.}
	
	We (1) assign $e_t$ a weight, $w( e_t)$, based on the reservoir and (2) compute the rank of edge $e_t$, i.e., $r( e_t)$, based on $w( e_t)$ and a random value $u$ randomly sampled in $(0, 1]$.
	\begin{itemize}
		\item \textbf{Case 1.1: $r( e_t) > \tau_p$.} We include edge $e_t$ in the reservoir. 
		\item \textbf{Case 1.2: $r( e_t) \le \tau_p$.} We discard edge $e_t$.
		
	\end{itemize}

	\item \textbf{Case 2: For an event $(+, e_t)$ with a \emph{full} reservoir.} 
	We {\chengfinal do (1) and (2) as in Case 1,}
	and (3) update $\tau_p$ to be minimum rank of the edges in the reservoir.
	\begin{itemize}
		\item \textbf{Case 2.1: $r(e_t) > \tau_{p}$.} We exclude the edge with the minimum rank,  include edge $e_t$ in the reservoir, and update $\tau_q$ to be $\tau_p$.
		\item \textbf{Case 2.2: $\tau_q <  r( e_t) \le \tau_p$.} We discard edge $e_t$, and update $\tau_q$ to be $r( e_t)$.
		\item \textbf{Case 2.3: $r( e_t) \le \tau_q$.} We discard the edge $e_t$.
	
	\end{itemize}

	\item \textbf{Case 3: For an event $(-, e_t)$.} We drop $e_t$ from the reservoir  if $e_t$ has been sampled before and do nothing otherwise.
\end{itemize}
The pseudo-code of \texttt{WSD} is presented in Algorithm~\ref{alg:sampling}.
%
We explain some of the intuitions behind the sampling process of \texttt{WSD}. \underline{First}, in Case 1, we do not update $\tau_p$ or $\tau_q$, which we explain as follows. In this case, the reservoir is not full, and thus no edges that have been included in the reservoir would be dropped. This implies that the probabilities of including these edges in the reservoir would not be changed, which further implies that $\tau_q$ should be retained since we aim to use $\tau_q$ for computing these probabilities, as shown in Eq.~(\ref{eq:prob-new}). 
Correspondingly, the rank threshold $\tau_p$ which has been used for sampling previous edges should be retained. 
%

To see this, we consider the same scenario {\chengfinal described} in Section~\ref{subsec:gps}. Let $p$ be the probability that an edge is included in the reservoir at time $t$. 
Let $\tau$ be the value of $\tau_p$ observed at $t-1$.
Let $\tau'$ (resp. $\tau''$) be the value of $\tau_p$ observed at $t'$ if the reservoir is non-full (resp. full). 
Let $p(\tau)$, $p(\tau')$ and $p(\tau'')$ be the probabilities that an edge's rank is no less than $\tau$, $\tau'$ and $\tau''$, respectively.
%
Consider an edge $e$ before $e_{t'}$. The probability that $e$ is included in the reservoir at $t'$ is $p\cdot p(\tau) + (1 - p) \cdot p(\tau'')$. 
For the term $p\cdot p(\tau)$, it corresponds to case that $e_{t}$ has been included in the reservoir at $t$ (which has the probability of $p$), and $p(\tau)$ corresponds to the conditional probability that $e$ is in the reservoir at $t-1$ (in this case, $e$ would be in the reservoir at $t'$ for sure since the reservoir is not full at $t'$). 
For the term $(1-p) \cdot p(\tau'')$, it corresponds the case that $e_{t}$ has not been included in the reservoir at $t$ (which has the probability of $(1-p)$, and in this case the reservoir is full at the beginning of $t'$ and $p(\tau'')$ corresponds to the conditional probability that $e$ is in the reservoir at $t'$. 
Consider the edge $e_{t'}$. The probability that $e_{t'}$ is included in the reservoir at $t'$ is $p \cdot p(\tau') + (1-p) \cdot p(\tau'')$, which can be verified similarly. 
%
%
%
{\CHENG{In order to guarantee that the probabilities for $e$ and $e_{t'}$ to be included are equal, it is necessary to have $\tau' = \tau$ (i.e., $\tau_p$ is retained). }}
%

\underline{Second}, in Case 2, we update $\tau_p$ to be the minimum rank of edges in the reservoir since the reservoir is full and for any edge $e_t$ to be included in the reservoir, its rank $r(e_t)$ should be larger than the minimum rank of existing edges in the reservoir. In addition, depending on the rank of $e_t$, we update the reservoir and $\tau_p$ correspondingly. In particular, in Case 2.1 and Case 2.2, we update $\tau_p$ appropriately so as to make sure Eq.~(\ref{eq:prob-new}) holds. 

\underline{Third}, in Case 3, we do not update $\tau_p$ or $\tau_q$ since after an edge is deleted, the probability that a remaining edge that has been inserted is included in the reservoir would not be affected. 



\begin{lemma}
\label{lemma:prob}
In \texttt{WSD}, at the end of a time $t$, the probability that an edge $e$ is sampled in the reservoir, i.e.,
$\mathbb{P}[e \in \mathcal{R}^{(t)}]$, is equal to $\mathbb{P}[r( e) > \tau_q]$, where $\tau_q$ is observed at time $t$. 
\end{lemma}

\begin{proof}
We prove this lemma by deduction. 
When $t \le M$, $\tau_q$ is always 0. Each edge would be included into reservoir for sure, which implies that Eq.~(\ref{eq:prob-new}) holds. 
Assume that Eq.~(\ref{eq:prob-new}) holds for $t = k$ $(k \ge M)$, i.e., $\mathbb{P}[e \in \mathcal{R}^{(k)}] = \mathbb{P}[r( e) > \tau_q^{(k)}]$, 
{\CHENG{where $\mathcal{R}^{(k)}$ and $\tau_q^{(k)}$ correspond to $\mathcal{R}$ and $\tau_q$ as observed at time $k$, respectively.}}
Consider $t = k + 1$. 
%
There are three possible cases of how the reservoir may change (cases as defined before). 
\begin{itemize}
    %
    \item \textbf{Case 1.} 
    Consider an edge $e$ {\CHENG{(before $e_{k+1}$)}}. Since (1) no edges that have been included in the reservoir would be dropped, which implies that the probabilities of including these edges in the reservoir would not be changed, and (2) we hold $\tau_q$, we have $\mathbb{P}[e \in \mathcal{R}^{(k+1)}] = \mathbb{P}[e\in \mathcal{R}^{(k)}] =  \mathbb{P}[r( e) > \tau_q^{(k)}] =  \mathbb{P}[r( e) > \tau_q^{(k+1)}]$. 
    Consider $e_{k+1}$. Since we hold $\tau_p$, we would sample $e_{k+1}$ with the exactly the same probability as we sample those  edges before $e_{k+1}$ with equal weights. Therefore, Eq.~(\ref{eq:prob-new}) holds.
    
    \item \textbf{Case 2.}
    Consider $e_{k+1}$. We compare $r(e_{k+1})$ with $\tau_p$, which corresponds to the minimum rank of the edges in {\CHENG{$\mathcal{R}^{(k)}$}}. In case that $e_{k+1}$ is included into the reservoir, we update $\tau_q^{(k+1)}$ to $\tau_p$. Therefore, at the end of time $t=k+1$, $\mathbb{P}[e_{k+1} \in \mathcal{R}^{(k+1)}] = \mathbb{P}[r( e_{k+1}) > \tau_p] = \mathbb{P}[r( e_{k+1}) > \tau_q^{(k+1)}]$. 
    %
    Consider an edge $e$ {\CHENG{(before $e_{k+1}$)}}. It can be verified that at the end of time $t$, 
    {\CHENG{$e$ is included in the reservoir}} only if its rank is larger than $\tau_q$, i.e., $\mathbb{P}[e \in \mathcal{R}^{(k+1)}] = \mathbb{P}[r(e) > \tau_q]$. Thus, Eq.~(\ref{eq:prob-new}) holds. 
    
    \item \textbf{Case 3. }
    Since (1) the deletion of $e_{k+1}$ would not affect other edges' probabilities, and (2) we do not update $\tau_q$. Consider an edge $e$ {\CHENG{(before $e_{k+1}$)}}. We have $\mathbb{P}[e \in \mathcal{R}^{(k+1)}] = \mathbb{P}[e\in \mathcal{R}^{(k)}] =  \mathbb{P}[r( e) > \tau_q^{(k)}] =  \mathbb{P}[r( e) > \tau_q^{(k+1)}]$. That is, Eq.~(\ref{eq:prob-new}) holds.
\end{itemize}
Therefore, $\forall t$, Eq~(\ref{eq:prob-new}) holds. 
\end{proof}
%
%


\smallskip
\noindent\textbf{Estimator and Analysis.}
%
Let $J=\{e_{i_1}, e_{i_2}, \cdots, e_{i_{|H|}}\}$ be a subgraph formed at time $t_a(J)$. Note that $t_a(J)$ is the time at which the last edge of $J$ appears, i.e., $t_a(J)=i_{|H|}$. We define a random variable $X_{\texttt{WSD}}^{J}$ for $J$ as follows.
\begin{equation}
    X_{\texttt{WSD}}^J= \displaystyle\prod \limits_{e \in J \setminus e_{i_{|H|}}} \frac{\mathbb{I}(e\in \mathcal{R})}{\mathbb{P}[r( e) > \tau_q]}
    \label{eq:est-1}
\end{equation}
where $\mathcal{R}$ and $\tau_q$ are the reservoir and the rank value, as observed just after time $t_a(J)-1$. 
%
Assume that a subgraph $J$ is destroyed at time $t_d(J)$ when a deletion event happens on an edge $e_x$ of $J$. We define another random variable $Y_{\texttt{WSD}}^J$ for $J$ as follows. 
\begin{equation}
    Y_{\texttt{WSD}}^J= \displaystyle\prod \limits_{e \in J \setminus e_x} \frac{\mathbb{I}(e\in\mathcal{R})}{\mathbb{P}[r( e) > \tau_q]}
    \label{eq:est-2}
\end{equation}
where $\mathcal{R}$ and $\tau_q$ are the reservoir and the rank value, as observed just after time $t_d(J)-1$. 
Based on these two subgraph estimators, we define an estimator of the count of subgraph structures at any time $t$ as follows.
\begin{equation}
    c_{\texttt{WSD}}^{(t)} = \sum_{J \in \mathcal{A}^{(t)}} X_{\texttt{WSD}}^J - \sum_{J \in \mathcal{D}^{(t)}} Y_{\texttt{WSD}}^J
\end{equation}
where $\mathcal{A}^{(t)}$ (resp. $\mathcal{D}^{(t)}$) is set of subgraphs which are isomorphic to $H$ and have been added to (resp. deleted from) the graph $G^{(t)}$ by time $t$.

\begin{theorem}[Unbiasedness of the subgraph count estimator of \texttt{WSD}] 
\label{theo:wsd}
Given the graph stream $S$ and $M\geq |H|$, 
$\forall t$,
we have
\begin{equation}
    \mathbb{E}[c_{\texttt{WSD}}^{(t)}] = |\mathcal{J}^{(t)}|.
\end{equation}
\end{theorem}

{\KXREVIEW{
We provide a sketch of the proof here and put the details in Appendix~\ref{app:proof-wsd}.  We first prove by deduction that at the end of time $t$ for all cases, the probability that a set of edges $E=\{e_1, \cdots, e_{|E|}\}$ $(|E| \le M)$ is included in $\mathcal{R}^{(t)}$ is
\begin{equation}
\label{eq:set}
    \mathbb{P}[E\subset \mathcal{R}^{(t)}] = \displaystyle\prod\limits_{e\in E} \mathbb{P}[e\in \mathcal{R}^{(t)}] = \displaystyle\prod\limits_{e\in E} \mathbb{P}[r(e)>\tau_q], 
\end{equation}
where $\tau_q$ is observed at time $t$. 
Then, we prove the following two statements based on Eq~(\ref{eq:set}). 
\begin{equation}
    \mathbb{E}[X_{\texttt{WSD}}^{J}] = 1, \quad \mathbb{E}[Y_{\texttt{WSD}}^{J}] = 1. 
\end{equation}
Finally, based on the linearity of expectation, we have 
\begin{equation}
\begin{aligned}
    & \mathbb{E}[c_{\texttt{WSD}}^{(t)}] = \sum_{J\in \mathcal{A}^{(t)}} \mathbb{E}[X_{\texttt{WSD}}^{J}] - \sum_{J\in \mathcal{D}^{(t)}} \mathbb{E}[Y_{\texttt{WSD}}^{J}] \\
    &= \sum_{J \in \mathcal{A}^{(t)}} 1 - \sum_{J \in \mathcal{D}^{(t)}} 1 
    = |\mathcal{A}^{(t)}| - |\mathcal{D}^{(t)}| =  |\mathcal{J}^{(t)}|.
\end{aligned}
\end{equation}
}}
\smallskip
\noindent\textbf{Implementation of the estimator.}
The subgraph counting on fully dynamic graph streams with \texttt{WSD} is shown in Algorithm~\ref{alg:counting}. 
We initialize the counter $c$ as 0 at the beginning. When an edge insertion (resp. deletion) event happens, we check whether the edge forms some subgraph structures with the sampled edges. If yes, we add (resp. subtract) the product of the inverse probabilities of these sampled edges, as defined in Eq.~(\ref{eq:est-1}) (resp. Eq.~(\ref{eq:est-2})), to (resp. from) the counter. Then we update the reservoir based on \texttt{insert} (resp. \texttt{delete}) function defined in Algorithm~\ref{alg:sampling}. 

\begin{algorithm}[t]
\small
\DontPrintSemicolon
\SetNoFillComment
\caption{Subgraph Counting with \texttt{WSD}}
\label{alg:counting}
\KwIn{An edge event stream $S$ {\KXREVIEW{and a subgraph pattern $H$}}.}
\KwOut{{\KXREVIEW{
Estimated subgraph count $c$}}. }
\SetKwFunction{insert}{insert}\SetKwFunction{delete}{delete}
Let $\mathcal{R}$ be a priority queue with the maximum size $M$\;
$\mathcal{R} \gets \Phi$, $\tau_p\gets 0$, $\tau_q \gets 0$, $c\gets 0$\;
\ForEach {$(op, e_t) \in S$} {

    $\mathcal{H} \gets \{J \subset (\mathcal{R} \cup e_t) \mid e_t\in J, J \cong H\}$\;
    \ForEach{$J \in \mathcal{H}$} {
        \uIf {$op = +$} {
            $c\gets c+\displaystyle\prod \limits_{e \in J \setminus e_t} \frac{1}{\mathbb{P}[r(e)>\tau_q]}$\;
        }
        \Else {
            $c\gets c-\displaystyle\prod \limits_{e \in J \setminus e_t} \frac{1}{\mathbb{P}[r(e)>\tau_q]}$\;
        }

    }
    \uIf {$op = +$} {
        \insert{$e_t$}\;
    }
    \Else {
        \delete{$e_t$}\;
    }
}
\Return $c$\;
\end{algorithm}



{\KXREVIEW{
\begin{theorem}[Complexities of \texttt{WSD}]
\label{theo:wsd-complexity}
The time and space complexity of \texttt{WSD} is $O((|A|+|D|) \cdot \log M \cdot \gamma(M))$ and $O(M)$, respectively, where all notations have the same meanings as those in Theorem~\ref{theo:gps-a-complexity}. 
\end{theorem}

\begin{proof}
{\chengfinal They can be verified similarly as those of \texttt{GPS-A}.}
\end{proof}
}}

{\KXREVIEW{
\smallskip\noindent
\textbf{Remarks.} (1) \texttt{WSD} v.s. \texttt{GPS-A}. While they have the same time complexity, \texttt{WSD} would run slightly faster than \texttt{GPS-A} in practice. 
This is because the reservoir of \texttt{GPS-A} is always full when $t\ge M$ while the reservoir of \texttt{WSD} could be smaller than $M$ after {\chengfinal some} deletion events happen. For  subsequent insertion events, it would take more time to add an edge to the reservoir of \texttt{GPS-A} than to that of \texttt{WSD}. 
(2) \texttt{WSD} v.s. Existing algorithms \texttt{Triest} \cite{stefani2017triest}, \texttt{ThinkD} \cite{shin2018think} and \texttt{WRS} \cite{lee2020temporal}. 
Existing algorithms all have a time complexity of $O((|A|+|D|\cdot M) \cdot \gamma(M))$. Thus, when the number of deletions satisfies 
\begin{equation}
    |D| > \frac{\log M-1}{M- \log M} \cdot |A| \approx \frac{\log M}{M} \cdot |A|,
\end{equation}
\texttt{WSD} would run faster than existing algorithms. 
}}

\section{Reinforcement Learning based Weight Function}
\label{sec:rl}


Recall that in the \texttt{WSD} algorithm, one step is to set the weight of an edge $e$ when it is inserted given the current reservoir $\mathcal{R}$. This step is modelled as a weight function $W(e, \mathcal{R})$. 
Due to the online nature and the one-pass sampling setting, the weights cannot be set optimally for optimizing some objective (e.g., minimizing the variance of the estimator). A common practice is to set the weights using some heuristics. For example, one heuristic is to use the number of subgraph structures that would be newly formed by the new edge as an indicator of the importance of the edge - the larger the number is, the higher the weight is~\cite{ahmed2017sampling}. While this heuristic is intuitive enough, it cannot be adaptive to different datasets and/or the underlying dynamics of edges. According to our experimental results, this heuristic based cannot work stably when the order of edge insertions/deletions is changed. 

We observe that whenever an edge is inserted, \texttt{WSD} needs to decide the weight of the edge. In addition, the weights decided for different edges would collectively affect the sampling process (i.e., the probabilities that the edges are sampled). This naturally {\chengfinal triggers} us to think of the reinforcement learning (RL) for this task of deciding the weights during the sampling process.
Note that RL is well known for effectively making sequential decisions so as to optimize an objective that is collectively affected by the decisions made at each step~\cite{puterman2014markov}. Next, we present our RL based method for this task by formalizing it as {\chengfinal an} Markov decision process (MDP) \cite{puterman2014markov} in Section~\ref{subsec:mdp} and then presenting the algorithm for learning the policy based on the MDP (which would be then used for deciding the weights of edges during the {\chengfinal sampling} process) {\chengfinal in Section~\ref{subsec:policy-learning}}. 

\subsection{Weight Function Modeled as an MDP}
\label{subsec:mdp}
Let $t_1, t_2, \cdots$ be the time steps, at which the edge insertion events happen. 
We define {\chengfinal an} MDP, which consists of (1) states, (2) actions, (3) transitions and (4) rewards, as follows.

\smallskip
\noindent
\textbf{(1) States}. We denote the state at time $t_k$ by $s_k$. Intuitively, the state $s_k$ should capture essential information of the edge in the whole graph $G^{(t_k)}$. Since we have no access to the edges which have not been sampled in the reservoir, we rely on the edges in the reservoir $\mathcal{R}$ and the new edge $e$. We identify two types of information, namely the topological information and the temporal information, for defining the state $s_k$. 

To capture the topological information of the new edge $e=(u,v)$, we take  the number of subgraph patterns which are newly formed by $e$ and some edges sampled in the reservoir, denoted by $|\mathcal{H}_k|$ (which is computed in line 4 of Algorithm~\ref{alg:counting}). Besides, we take the number of the neighbors of $u$ (resp. $v$) in the sampled graph, denoted by $|\mathcal{N}_k(u)|$ (resp. $|\mathcal{N}_k(v)|$). We concatenate these values together and form a vector, denoted by $s_k^g$, that is, 
\begin{equation}
    s_k^g = [|\mathcal{H}_k|, |\mathcal{N}_k(u)|, |\mathcal{N}_k(v)|]. 
\end{equation}
The rationale of $s_k^g$ is as follows. The number of the subgraphs $|\mathcal{H}_k|$ indicates the importance of the edge $e$ to some extent - if the edge $e$ can form more subgraphs in the sampled graph, it has a larger probability to form more subgraphs in the complete graph. The numbers of neighbors ($|\mathcal{N}_k(u)|$ and $|\mathcal{N}_k(v)|$) indicate the potential of the edge $e$ for forming subgraphs in the future. For example, in the example discussed in Section~\ref{sec:intro}, famous people can form more triangles than the general public since they have more connections with others. For other subgraphs, similarly, the edges that are adjacent to more others tend to have higher chance to form subgraphs.

To capture the temporal information of the new edge $e=(u,v)$, we consider all subgraphs that are formed by the edge $e$, i.e., $\mathcal{H}_k$. For each such subgraph, it is associated with an ordered set of edges, i.e., $J=\{e_{i_1}, e_{i_2}, \cdots, e_{i_{|H|}}\}$, where $i_1<i_2<\cdots<i_{|H|}=t_k$. Note that $e_{i_{|H|}}=e$. 
Then for each entry $j\in[1,|H|]$, we take the maximum value $i_{j}$ among these $|\mathcal{H}_k|$ subgraphs and denote it by $v_j$, that is,
\begin{equation}
\label{eq:vj}
    v_j = \max\{i_j \mid e_{i_j} \in J, J\in \mathcal{H}_k \}.
\end{equation}
Finally, we concatenate these values $v_j$ for $j\in [1, |H|]$ and form a vector, denoted by $s_k^v$, that is
\begin{equation}
\label{eq:sv}
    s_k^v = [v_1, v_2, \cdots, v_{|H|}]. 
\end{equation}
The rationale of $s_k^v$ is that it provides the temporal information related to edge $e$. For example, if the elements of $s_k^v$ are large, we know that the subgraphs are more likely to be formed with recent edges, and thus sampling edge $e$ with a larger probability may help form more subgraphs in the near future.


In summary, we define the state $s_k$ to be a $(|H|+3)$-dimension vector, which captures both topological and temporal information, as follows. 
\begin{equation}
    s_k = [s_k^g, s_k^v] \in \mathbb{R}^{|H|+3}. 
\end{equation}

\smallskip
\noindent
\textbf{(2) Actions}. We denote the action at time $t_k$ by $a_k$. Recall that at time $t_k$, the state is $s_k$, and we need to assign a weight for the edge $e$ to represent its importance. Formally, we define $a_k$ as follows.
\begin{equation}
    a_k = w \quad (w>0), 
\end{equation}
which means we assign a positive real weight $w$ to the edge, i.e., $W(e, \mathcal{R}) = w$.

\smallskip
\noindent
\textbf{(3) Transitions}. Let $s_k$ be a state and suppose that we take an action $a_k = w$, which assigns the edge $e$ a weight $w$. Then the \texttt{WSD} algorithm would proceed with the weight of $e$ until a new edge insertion arrives. Then, we arrive a new state and compute the state $s_{k+1}$ accordingly.

\smallskip
\noindent
\textbf{(4) Rewards}. We first define the error between the estimation and the ground truth at time $t$, denoted by $\epsilon(t)$, as follows.
\begin{equation}
    \epsilon(t) = |c^{(t)} - |\mathcal{J}^{(t)}||, 
\end{equation}
where $c^{(t)}$ is the value of the subgraph estimation and $|\mathcal{J}^{(t)}|$ is the value of the ground truth of the subgraph counts. 
Consider that we take action $a_k$ at a state $s_k$ and then we arrive at a new state $s_{k+1}$. We define the reward associated with this transition from $s_k$ to $s_{k+1}$, denoted by $r_k$, as follows, 
\begin{equation}
    r_k = \epsilon(t_k) - \epsilon(t_{k+1}), 
\end{equation}
The intuition is that if the error between the estimation and the ground truth resulted from the reservoir after processing the edges between $t_k$ and $t_{k+1}$, i.e., $\epsilon(t_{k+1})$, is smaller, then the reward is larger. It is worthy of noting that with the reward defined as above, the objective of the MDP, which is to maximize the accumulative rewards, would be equivalent to that of the subgraph counting problem, which is to minimize the estimation error. To see this, suppose that we go through a sequence of states $s_1, s_2, \cdots, s_N$ and correspondingly we receive a sequence of rewards $r_1, r_2, \cdots, r_{N-1}$. In the case that the future rewards are not discounted, we have
\begin{equation}
    \sum_{i=1}^{N-1} r_i = \sum_{i=1}^{N-1} (\epsilon(t_i) - \epsilon(t_{i+1})) = \epsilon(t_1) - \epsilon(t_{N}) = - \epsilon(t_{N})
\end{equation}
Note that $\epsilon(t_1)=0$ because at the very beginning, the estimation and the ground truth are both 0, and $\epsilon(t_N)$ corresponds to the error of the subgraph counts at time $t_N$.

\subsection{Policy learning}
\label{subsec:policy-learning}


The core problem of an MDP is to find an optimal policy for the agent, which corresponds to a function that specifies the action that the agent should choose in a specific state so as to maximize the accumulative rewards.
{\KXREVIEW{In our MDP, the states are high dimensional vectors, and the actions are in a continuous domain. Therefore, we {\CHENG adopt} DDPG \cite{lillicrap2016continuous}, an actor-critic method, to solve our MDP as it targets the same setting. Specifically, DDPG {\CHENG maintains} two {\CHENG main} networks. One is an actor network $\mu(s;\theta)$, parameterized by $\theta$, which specifies the policy by deterministically mapping a state to a specific action,}} 
\begin{equation}
    a = \mu(s; \theta) = \sigma(\bm{W} s+\bm{b})
\end{equation}
where $\theta$ denotes the parameters $\{\bm{W}, \bm{b}\}$, and $\sigma$ denotes the activation function. 
{\KXREVIEW{
The other is a critic network $Q(s,a; \phi)$, parameterized by $\phi$, which approximates the expected accumulative rewards the agent would receive by following any policy after seeing state $s$ and taking action $a$.}}

For the training process, DDPG initializes two main networks $\mu(s; \theta)$ and $Q(s,a; \phi)$.
It also involves two target networks $\mu'(s; \theta')$ and $Q'(s,a; \phi')$, which are used for calculating the losses for training the main networks. 
Besides, {\CHENG it} maintains a replay memory, which contains the transitions that are used for training the network. The training process is as follows. 
Consider $N$ experiences sampled uniformly from the replay memory, i.e., $(s_i, a_i, r_i, s_{i+1})$ for $i=1, 2, \cdots, N$. 
{\KXREVIEW{For the critic network $Q(s,a;\phi)$, {\CHENG it} computes the loss using the Bellman equation \cite{mnih2015human}, i.e., }}
\begin{equation}
    L(\phi) = \frac{1}{N} \sum_{i=1}^N (y_i - Q(s_i, a_i; \phi))^2
\end{equation}
\begin{equation}
    y_i = r_i + \gamma \cdot Q'(s_{i+1}, \mu'(s_{i+1}; \theta'); \phi')
\end{equation}
where $\gamma$ is a discount factor.
{\KXREVIEW{For the actor network $\mu(s;\theta)$, 
{\CHENG it computes a loss as the negation of the expected return,}
i.e.,
\begin{equation}
    L(\theta) = -\frac{1}{N} \sum_{i=1}^N Q(s_i, \mu(s_i; \theta); \phi)
\end{equation}
}}
Finally, it updates the parameters $\theta$ and $\phi$ by gradient descent.


\section{Experiments}
\label{sec:exp}

\begin{table}[t]
\scriptsize
    \centering
    \caption{Dataset statistics.}
    \vspace{-2mm}
    \label{tab:data}
    \begin{tabular}{c|cc|cc}
    \toprule
          Category & {\KXREVIEW{Graph (Train)}} & {\KXREVIEW{$|E|$}} & Graph (Test)  & $|E|$ \\
    \midrule
        Citation & {\KXREVIEW{cit-HepTH (HE)}} & {\KXREVIEW{2.67M}} & cit-patent (PT) & 16.5M\\
        Community & {\KXREVIEW{com-DBLP (DB)}} & {\KXREVIEW{1.04M}} & com-youtube (YT) & 2.99M \\
        Social & {\KXREVIEW{soc-Texas84 (TX)}} & {\KXREVIEW{1.59M}} & {\KXREVIEW{soc-twitter (TW)}} & {\KXREVIEW{265M}} \\
        Web & {\KXREVIEW{web-Stanford (SF)}} & {\KXREVIEW{2.31M}} & web-google (GL) & 5.10M \\
    \bottomrule
    \end{tabular}
    \vspace{-4mm}
\end{table}


\subsection{Experimental Setup}
\label{subsec:setp-up}

\noindent
\textbf{Datasets}.
We use both real datasets and synthetic datasets in our experiments.
We use four types of real graphs, including citation graphs, web graphs, community networks and online social networks, which are collected from an open-source network repository \cite{rossi2015network, rossi2016interactive}. 
For each graph in a dataset, we ignore the directions, weights and self-loops (if any).
There are two ways to generate deletions for constructing fully dynamic graph streams, namely (1) massive deletion \cite{stefani2017triest} and (2) light deletion \cite{lee2020temporal}, respectively. 
We generate fully dynamic graph streams under the massive deletion scenario as follows. All edges are inserted in their natural order but each edge insertion is followed with probability $\alpha$ by a massive deletion event where each edge currently in the graph is deleted with probability $\beta_m$ independently. 
We generate fully dynamic graph streams under the light deletion scenario as follows. 
All edges are inserted in their natural order and each edge has probability $\beta_l$ to be deleted, where the deletion is located at a random position after the corresponding edge insertion. 
%
%
Following the existing study \cite{lee2020temporal}, we generate some synthetic datasets by Forest Fire (FF) \cite{leskovec2007graph} model $G(n, p)$, in which $n$ represents the number of vertices and $p$ controls the density of the graph. Specifically, vertices arrive one at a time and form edges with some subset of the earlier vertices. Since the graph $G$ generated by FF maintains the structural properties (heavy-tailed degree distributions, communities, etc.) and temporal properties (densification, etc.) of the real graphs, we transform $G$ into an edge stream in their natural order. 
To ensure the synthetic datasets have similar scales as the real graphs, we set $n=2$M and $p=0.5$ based on empirical results.
%

\smallskip
\noindent
\textbf{Baselines and Metrics}. 
We consider two instances of the \texttt{WSD} algorithm, which adopt two different weight functions.
The first one uses the learned policy via reinforcement learning, which we denote by \texttt{WSD-L}.
The second one uses a weight function defined with an existing heuristic~\cite{ahmed2017sampling}, which we denote by \texttt{WSD-H}.
Specifically, in \texttt{WSD-H}, $W(e, \mathcal{R}) = 9 \cdot |\mathcal{H}(e)| + 1$ \cite{ahmed2017sampling}, where $|\mathcal{H}(e)|$ is the number of subgraphs completed by edge $e$ and some other edges in $\mathcal{R}$. 
%
We compare our \texttt{WSD} algorithms with three existing algorithms, namely \texttt{Triest} \cite{stefani2017triest}, \texttt{ThinkD} \cite{shin2018think} and \texttt{WRS} \cite{lee2020temporal}. 
We also compare our algorithm with the adapted framework \texttt{GPS-A}.
All algorithms are single-pass algorithms that estimate the triangle counts with a fixed budget for storing sampled edges. 
For the subgraph patterns, we first choose the triangle (3-clique) and the wedge (length-2 paths) by following the existing study~\cite{ahmed2017sampling} since these are the two mostly used subgraph structures.
In addition, we also consider a dense subgraph pattern, i.e., 4-clique, to verify that our algorithm can also work well on some large subgraph structures. 
We compare different algorithms by using the following metrics. 

\begin{itemize}
    \item \textbf{Absolute Relative Error (ARE)}. Let $X$ be the ground truth of the number of subgraphs at the end of the input streams and $\widehat{X}$ be an estimated value of $X$. The absolute relative error is $\frac{|\widehat{X}-X|}{X} \times 100\%$. 
    \item \textbf{Mean Absolute Relative Error (MARE)}. Let $X_t$ be the ground truth of the number of subgraphs by $t$ and $\widehat{X}_t$ be an estimated value of $X_t$. The mean absolute relative error is $\frac{1}{T}\sum_{i=1}^{T} \frac{|\widehat{X}_i-X_i|}{X_i} \times 100\%$. 
\end{itemize}
Due to the random nature of all algorithms, for all experiments, given a setting, we report the mean over 100 times of sampling.

\smallskip\noindent
\textbf{Policy Learning}. There are two neural networks used in \texttt{WSD-L} method. The actor network $\mu(s; \theta)$ involves one input layer and one output layer, and uses ReLU function as the activation function. We add one to the output to avoid assigning zero weights. The critic network $Q(s,a; \phi)$ involves one input layer, one hidden layer and one output layer, where the hidden layer involves 10 neurons and uses ReLU as the activation function. To avoid data scale issues, batch normalization is employed before the activation. 
For each graph in a {\KAIXIN{real}} dataset, we use another real graph under the same type for training, {\CHENGREVIEW as specified in Table~\ref{tab:data}}, since they may have some similar structural and temporal properties. For each graph in the synthetic dataset, we generate a graph $G(n=2\text{M}, p=0.5)$ by FF model for training.
{\KXREVIEW{
Under massive (resp. light) deletion scenario, we generate 10 different edge event streams with the same parameters $\alpha$ and $\beta_m$ (resp. $\beta_l$) and use these generated graphs for training based on empirical findings. 
Basically, using fewer streams would suffer from the over-fitting problem and using more would incur larger training cost without improving the quality of the model much.
}}
%
{\CHENGREVIEW We set the maximum size of the replay buffer as 10,000 and the parameter $N$ as 128.}
We train the networks for 1,000 iterations. In addition, we use the Adam stochastic gradient descent with learning rate of 0.001. For the reward discount $\gamma$, we set it as 0.99. 

\smallskip
\noindent
\textbf{Evaluation Platform}. 
All methods for comparison are implemented in C++. Note that we first implement and train \texttt{WSD-L} using Pytorch (Python 3.6). Then we hardcode the parameters $\theta=\{\bm{W}, \bm{b}\}$ in C++ to improve the efficiency. All experiments are conducted on a machine with Intel Core i9-10940X CPU and a single Nvidia GeForce 2080Ti GPU. Implementation codes and datasets can be found via this link \url{https://github.com/wangkaixin219/WSD/}.

\begin{table}[t]
\scriptsize
\centering
\caption{Results of counting wedges under the massive deletion scenario, where we use the default values of the parameters, $M=200,000$, $\alpha=3,000,000^{-1}$ and $\beta_m=0.8$.}
\vspace{-2mm}
\label{tab:res-wedge-massive}
\begin{tabular}{c|cccccc}
    \toprule
     \textbf{Graph} & \texttt{WSD-L} & \texttt{WSD-H} & \texttt{GPS-A} & \texttt{Triest} & \texttt{ThinkD} & \texttt{WRS} \\ 
    \midrule
    & \multicolumn{6}{c}{Absolute Relative Error ($\%$)} \\ \cmidrule{2-7}
    cit-PT & \textbf{0.046} & \underline{0.051} & 0.058 &  0.077 & 0.071 & 0.062 \\
    com-YT  & \textbf{0.011} & \underline{0.013} & 0.061 & 0.125 & 0.104 & 0.092 \\
    {\KXREVIEW{soc-TW}} & {\KXREVIEW{\textbf{0.243}}} & {\KXREVIEW{\underline{0.411}}} & {\KXREVIEW{0.434}} & {\KXREVIEW{0.627}} & {\KXREVIEW{0.572}} & {\KXREVIEW{0.483}}\\
    web-GL & \textbf{0.041} & \underline{0.044} & 0.117 & 0.815 & 0.670 & 0.366 \\
    synthetic  & \textbf{0.107} & \underline{0.148} & 0.192 & 0.564  & 0.324 & 0.231 \\
    \midrule
    & \multicolumn{6}{c}{Mean Absolute Relative Error ($\%$)} \\ \cmidrule{2-7}
    cit-PT & \textbf{0.036} & \underline{0.039} & 0.072 & 0.104 & 0.067 & 0.054 \\
    com-YT & \textbf{0.007} & \underline{0.008} & 0.057 & 0.061 & 0.054 & 0.044 \\
    {\KXREVIEW{soc-TW}} & {\KXREVIEW{\textbf{0.306}}} & {\KXREVIEW{\underline{0.451}}} & {\KXREVIEW{0.688}} & {\KXREVIEW{1.310}} & {\KXREVIEW{1.003}} & {\KXREVIEW{0.894}} \\
    web-GL & \textbf{0.039} & \underline{0.040} & 0.078 & 0.771 & 0.253 & 0.101 \\
    synthetic & \textbf{0.213} & \underline{0.312} & 0.401 & 0.646 & 0.769 & 0.443\\
    \midrule
    & \multicolumn{6}{c}{Running Time (s)} \\ \cmidrule{2-7}
    cit-PT & \underline{65.9} & \textbf{62.4} & 67.1 & 186.7 & 188.0 & 191.9 \\
    com-YT & 9.6 & 9.3 & 9.9 & \textbf{7.8} & \underline{7.9} & 8.1 \\
    {\KXREVIEW{soc-TW}} & {\KXREVIEW{\underline{1,940.7}}} & {\KXREVIEW{\textbf{1,832.2}}} & {\KXREVIEW{2,981.9}} & {\KXREVIEW{4,392.3}} & {\KXREVIEW{4,417.3}} & {\KXREVIEW{4,592.1}} \\
    web-GL & \underline{13.7} & \textbf{13.2} & 14.3 & 14.1 & 14.6 & 14.8 \\
    synthetic & \underline{43.3} & \textbf{42.9} & 45.2 & 262.7 & 263.2 & 267.9\\
    \bottomrule
\end{tabular}
\vspace{-2mm}
\end{table}

\begin{table}[t]
\scriptsize
\centering
\caption{Results of counting triangles under the massive deletion scenario, where we use the default values of the parameters, $M=200,000$, $\alpha=3,000,000^{-1}$ and $\beta_m=0.8$.}
\vspace{-2mm}
\label{tab:res-triangle-massive}
\begin{tabular}{r|cccccc}
    \toprule
    \textbf{Graph} & \texttt{WSD-L} & \texttt{WSD-H} & \texttt{GPS-A} & \texttt{Triest} & \texttt{ThinkD} & \texttt{WRS} \\ 
    \midrule
    & \multicolumn{6}{c}{Absolute Relative Error ($\%$)} \\ \cmidrule{2-7}
    cit-PT & \textbf{0.075} & \underline{0.083} & 0.106 & 0.175 & 0.143 & 0.142 \\
    com-YT & \textbf{0.048} & \underline{0.053} & 0.073 & 0.188 & 0.109 & 0.067 \\
    {\KXREVIEW{soc-TW}} & {\KXREVIEW{\textbf{0.404}}} & {\KXREVIEW{\underline{0.712}}} & {\KXREVIEW{0.893}} & {\KXREVIEW{1.214}} & {\KXREVIEW{1.056}} & {\KXREVIEW{0.952}}\\
    web-GL & \textbf{0.031} & \underline{0.037} & 0.734 & 0.197 & 0.195 & 0.136 \\
    synthetic & \textbf{2.507} & \underline{3.124} & 3.612 & 4.293 & 3.318 & 3.143 \\
    \midrule
    & \multicolumn{6}{c}{Mean Absolute Relative Error ($\%$)} \\ \cmidrule{2-7}
    cit-PT & \textbf{0.059} & \underline{0.065} & 0.083 & 0.113 & 0.089 & 0.079 \\
    com-YT & \textbf{0.052} & \underline{0.064} & 0.141 & 0.236 & 0.144 & 0.204 \\
    {\KXREVIEW{soc-TW}} & {\KXREVIEW{\textbf{0.513}}} & {\KXREVIEW{\underline{0.748}}} & {\KXREVIEW{0.946}} & {\KXREVIEW{1.433}} & {\KXREVIEW{1.341}} & {\KXREVIEW{1.005}}  \\
    web-GL & \textbf{0.054} & \underline{0.088} & 0.207 & 1.024 & 0.809 & 0.319 \\
    synthetic & \textbf{3.614} & \underline{4.278} & 4.662 & 6.691 & 5.643 & 4.817 \\
    \midrule
    & \multicolumn{6}{c}{Running Time (s)} \\ \cmidrule{2-7}
    cit-PT & \underline{70.4} & \textbf{66.7} & 71.5 & 189.5 & 193.1 & 197.4 \\
    com-YT & 10.2 & 9.9 & 10.6 & \textbf{8.0} & \underline{8.2} & 8.3 \\
    {\KXREVIEW{soc-TW}} & {\KXREVIEW{\underline{2,144.4}}} & {\KXREVIEW{\textbf{2,091.9}}} & {\KXREVIEW{3,329.3}} & {\KXREVIEW{4,810.8}} & {\KXREVIEW{4,839.3}} & {\KXREVIEW{5,023.5}} \\
    web-GL & \underline{16.1} & \textbf{14.9} & 17.3 & 16.9 & 17.5 & 18.4 \\
    synthetic & \underline{49.3} & \textbf{48.1} & 53.7 & 260.2 & 262.1 & 265.0\\
    \bottomrule
\end{tabular}
\vspace{-4mm}
\end{table}

\subsection{Experimental Results}
\label{subsec:exp-res}

\noindent
\textbf{(1) Evaluation on real and synthetic datasets}. 
{\KXREVIEW{Table~\ref{tab:res-wedge-massive} and Table~\ref{tab:res-triangle-massive} show the results of counting triangles and wedges under massive deletion scenario. The result of counting 4-cliques {\CHENGREVIEW and those under the light deletion show similar clues and are thus presented} in Appendix~\ref{app:exp}.}}
%
%
Consider the effectiveness. The results meet our expectations that weighted sampling works better than uniform sampling on all datasets and under two metrics. 
In addition, because \texttt{WSD-L} uses RL to adaptively capture the importance of the edges, it further improves the effectiveness of \texttt{WSD-H} which applies a heuristic based weight function. 
Consider the efficiency. As analyzed in Section~\ref{subsec:wsd}, for all datasets except com-YT, we observe that the number of generated deletions satisfies $|D|>{(\log M)}/{M} \cdot |A|$, where ${(\log M)}/{M}\approx 3\times 10^{-5}$, and thus our proposed algorithms run much faster than existing algorithms. In addition, because \texttt{WSD-L} needs to calculate a state when an insertion event comes, it runs slightly slower than \texttt{WSD-H}. 
Compared among different subgraph patterns, the ARE and MARE become larger as the size of the subgraph pattern increases. The reasons are as follows. First, when the size of the subgraph pattern increases, it becomes more difficult for a new edge to form some subgraph structures with the sampled edges. For example, when an edge event comes, a 4-clique could be formed only if the other 5 edges have been sampled while a wedge would be formed if the other edge has been sampled. Also, the inclusion probability of a large subgraph pattern is usually much smaller than that of a small subgraph pattern, which means a higher variance would be produced when counting larger subgraph structures. 
%

\begin{figure}[t]
    \centering
    {\KXREVIEW{
        \subfigure[ARE (massive deletion scenario)]{
	    \label{fig:scale-are-massive}
		\includegraphics[width=0.23\textwidth]{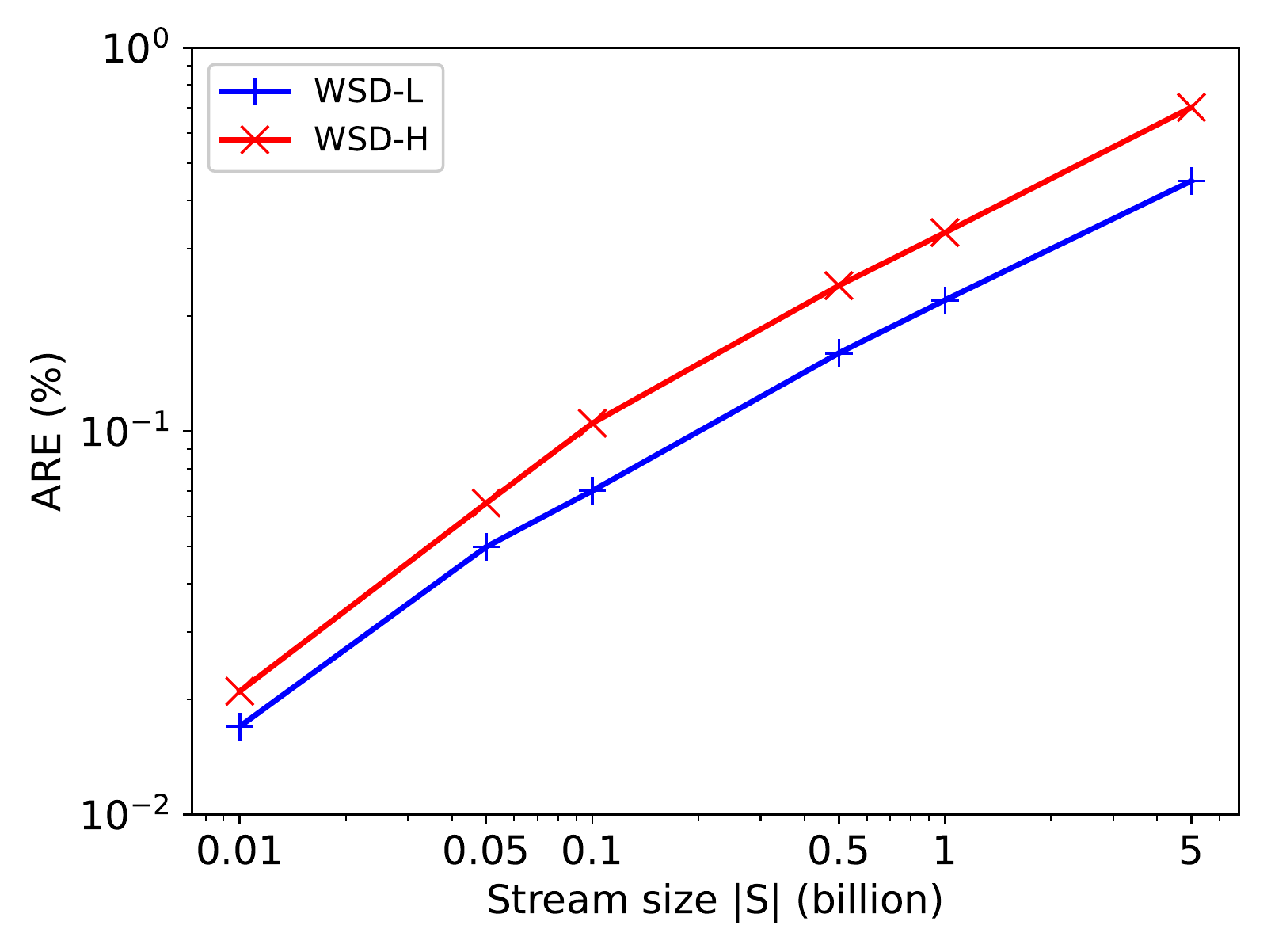}}
	\subfigure[Time (massive deletion scenario)]{
	    \label{fig:scale-time-massive}
		\includegraphics[width=0.23\textwidth]{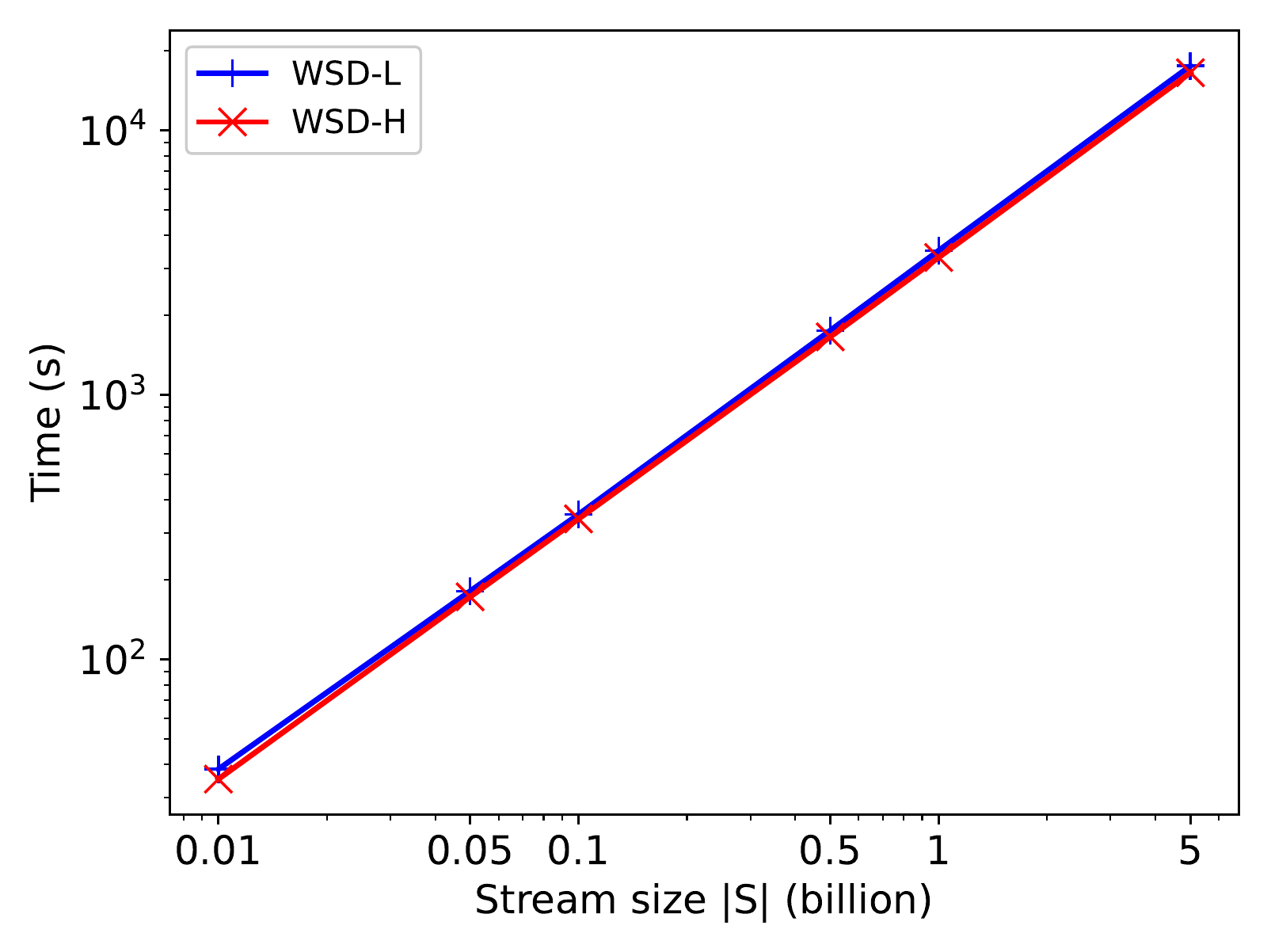}}
    \vspace{-2mm}
    \caption{Results on scalability test of counting triangles, reporting the ARE and running time of \texttt{WSD-L} and \texttt{WSD-H} with different stream sizes (synthetic). }
    \vspace{-4mm}
    \label{fig:scale-massive}
    }}
\end{figure}

{\KXREVIEW{
\smallskip\noindent
\textbf{(2) Scalability test. } 
%
To {\CHENGREVIEW study} the scalability of \texttt{WSD-L} and \texttt{WSD-H}, we generate a series of graphs with different sizes. To be specific, we first use FF model $G(n=1\text{B}, p=0.5)$ to generate the whole graph with slightly more than 5B edges (5.07 billion edges). We then create the edge event stream by adding deletions under massive and light deletion scenarios. Finally, we generate {\CHENGREVIEW graphs with different sizes} by picking the first 10M, 50M, 100M, 500M, 1B and 5B edge events.
We set $M$ as 1M following \cite{ahmed2017sampling}. 
The results under massive deletion scenario are shown in Figure~\ref{fig:scale-massive}. 
Consider the effectiveness. We observe that as the size of the stream increases, the ARE results also increases. The reason is that for all sizes of the streams, we use 1M edges as the sample. Therefore, the estimation results on a stream with more edge events would be less accurate. {\CHENGREVIEW Still}, we can provide a relatively accurate estimation ($<1\%$ error) when we sample only $0.02\%$ edges.  
Consider the efficiency. The running time of \texttt{WSD} is linear to the size of the stream, which is consistent with our theoretical analysis of the time complexity of \texttt{WSD}. Moreover, we also collect the results of the average time of updating the reservoir when an edge event occurs. The average processing time of each edge is around $3.2\mu$s. 

}}

\begin{figure*}[th]
    \centering
    \subfigure[Ordering of the stream (cit-PT)]{
	    \label{fig:massive-are-o}
		\includegraphics[width=0.235\textwidth]{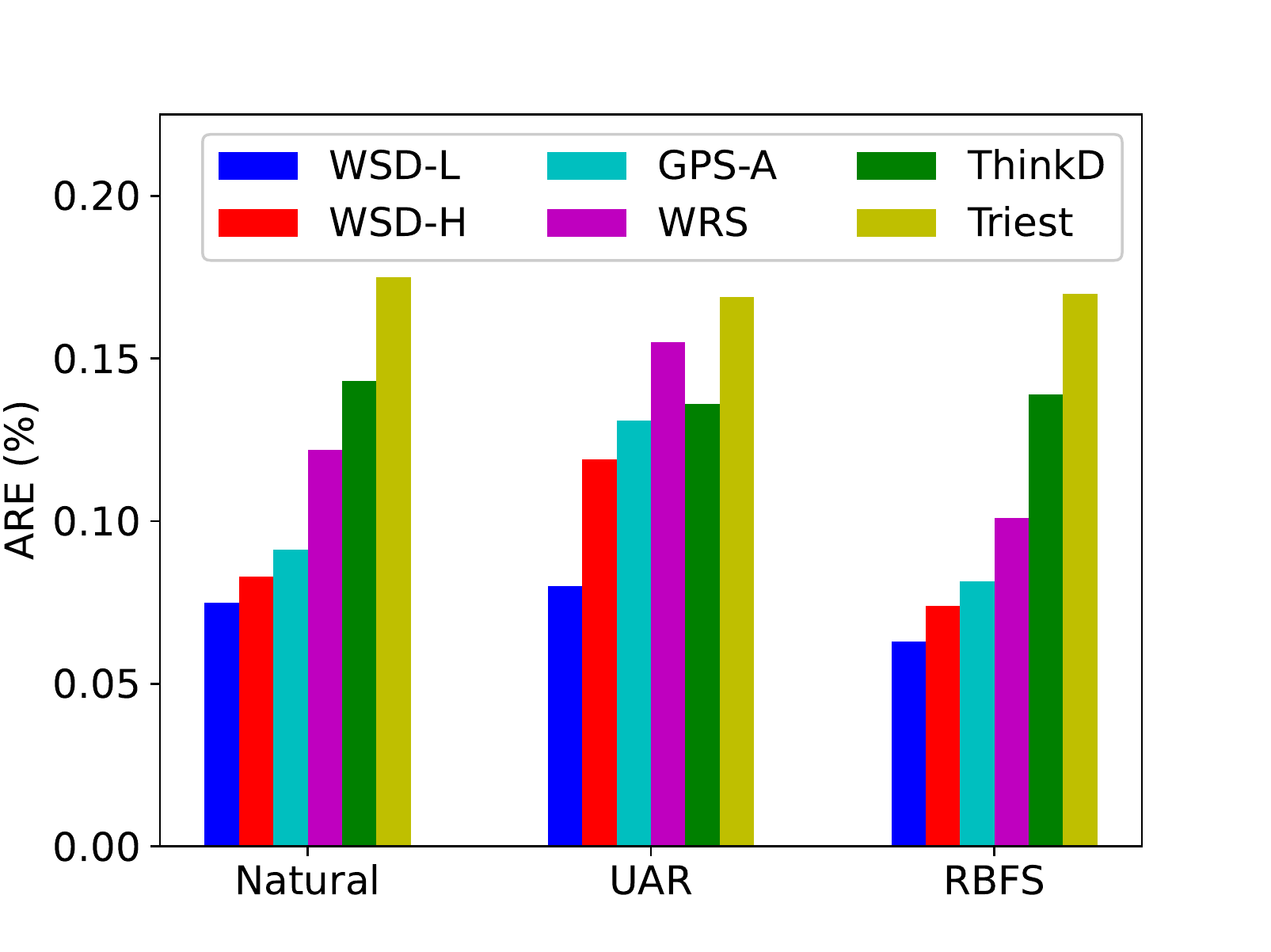}}
	\subfigure[Max. reservoir size $M$ (cit-PT)]{
        \label{fig:massive-are-m}
		\includegraphics[width=0.235\textwidth]{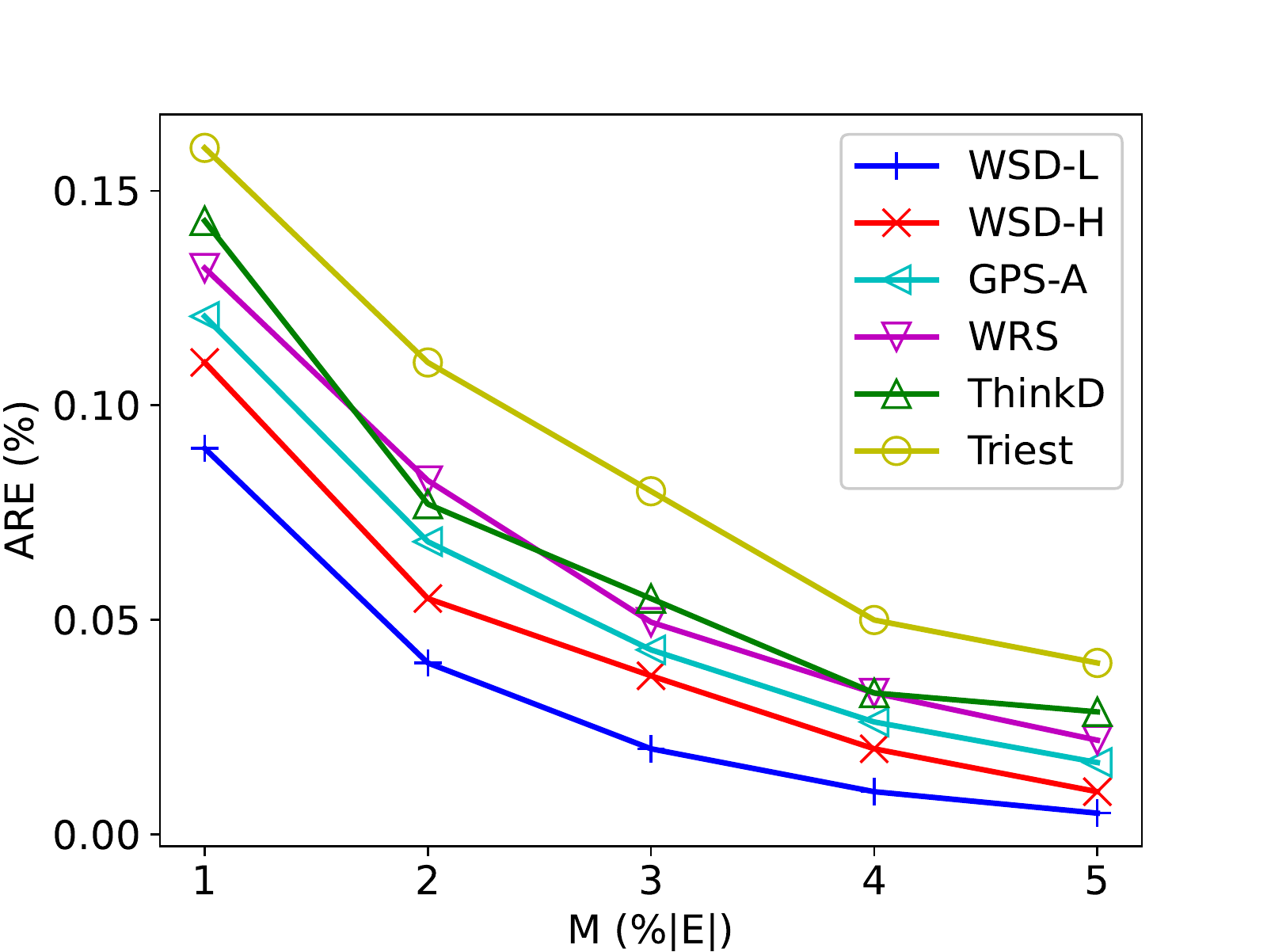}}
    \subfigure[{\KXREVIEW{Training size} (synthetic)}]{
	    \label{fig:train-massive}
		\includegraphics[width=0.22\textwidth]{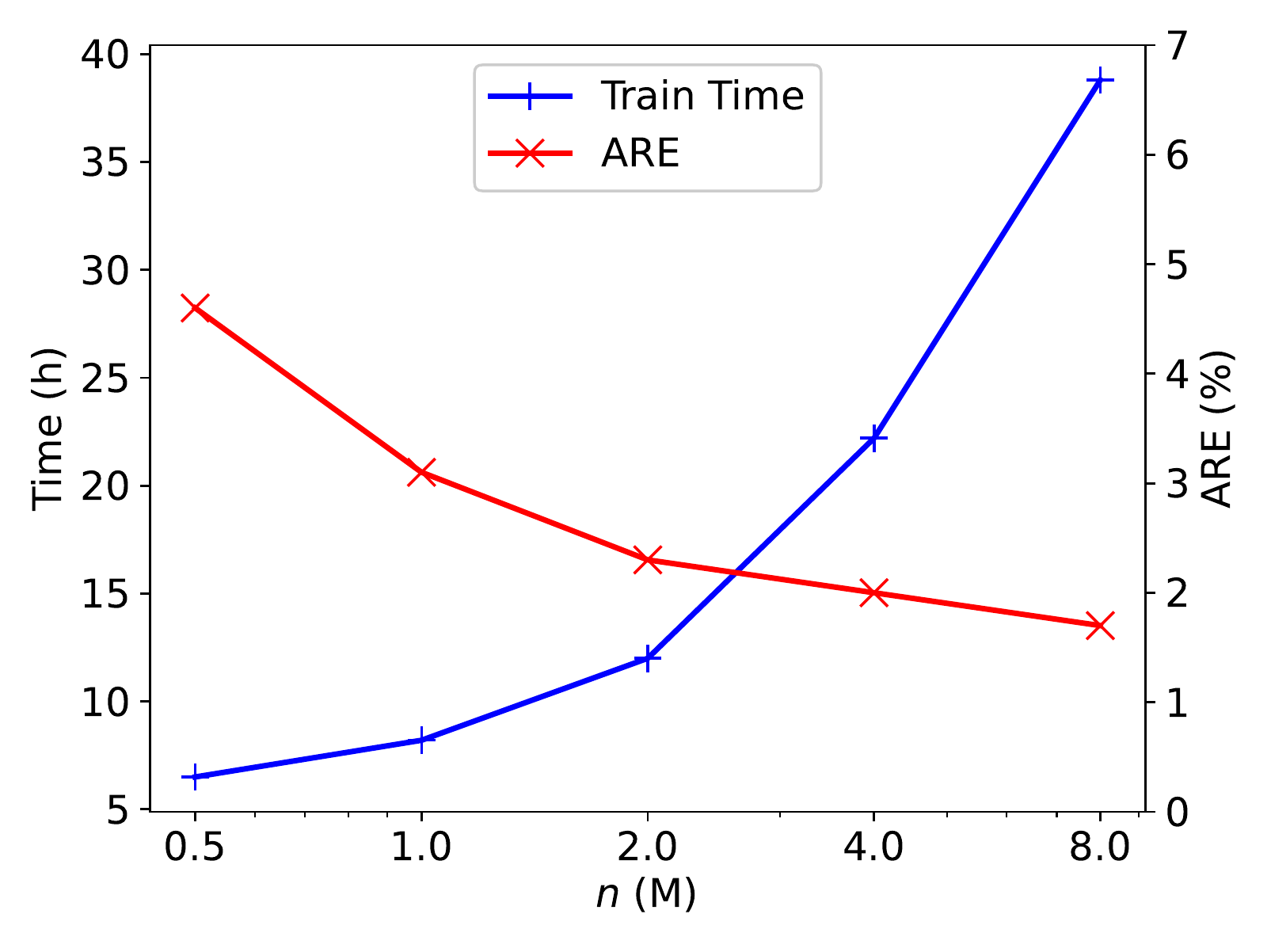}}
	\subfigure[\KXREVIEW{Relationship (cit-PT)}]{
	    \label{fig:relation-massive}
		\includegraphics[width=0.235\textwidth]{sup_figs/relation-massive.pdf}}
    \vspace{-2mm}
    \caption{Results of counting triangles under massive deletion scenario: (a) the ARE under different {\chengfinal orderings} of the stream; (b) the ARE under different {\chengfinal $M$'s}; (c) the training time and ARE under different {\chengfinal sizes} of training graphs; (d) the relationship between the weights and {\chengfinal numbers of triangles}. }
    \vspace{-4mm}
    \label{fig:res-massive}
\end{figure*}

\smallskip
\noindent
\textbf{(3) Impact of the ordering of the stream}. We follow \cite{stefani2017triest} to generate different orderings of the stream. The default setting is the natural ordering of a graph. We also consider the uniform-at-random (UAR) ordering and the random BFS (RBFS) ordering. The UAR ordering is to add edges in the order of a permutation of the natural order while the RBFS ordering is to start from a random vertex and add edges in the order of a BFS exploration of the graph. 
For example, in a social network, when a celebrity registers an account in a new platform, his/her followers would likely create connections (i.e., edge insertions) to him/her in a very short time, as BFS does. Figure~\ref{fig:massive-are-o} shows the ARE results of counting triangles on cit-PT under massive deletion scenario. 
%
We observe that \texttt{WSD-L} algorithm provides estimations with the smallest errors {\CHENG across different settings of ordering}.
These results show the robustness of our \texttt{WSD-L} algorithm which is due to its data-driven nature and adaptability to various underlying dynamics.

\smallskip\noindent
\textbf{(4) Effects of $M$}. We change the maximum reservoir size from $1\%\times |E|$ to $5\%\times |E|$. Figure~\ref{fig:massive-are-m} shows the results of counting triangles on cit-PT under the massive deletion scenario. 
%
The MARE results and the results of counting triangles on other datasets show the similar trends, thus omitted. Under both scenarios, we set the probability parameters same as those used in previous experiments. Our proposed algorithms consistently outperform other existing algorithms. Our algorithm can provide an accurate estimation ($<0.01\%$ error) when we sample only $4\%$ edges.

\begin{table}[h]
\vspace{-4mm}
    \centering
    {\KXREVIEW{
    \caption{Training time (hours) of counting triangles ($\triangle$) and wedges ($\wedge$) on four real datasets under massive deletion scenario. }
    \vspace{-2mm}
    \label{tab:train-massive}
    \begin{tabular}{c|c|c|c|c|c|c|c|c}
    \toprule
    \textbf{Graph} & \multicolumn{2}{c|}{cit-HE} & \multicolumn{2}{c|}{com-DB} & \multicolumn{2}{c|}{soc-TX} & \multicolumn{2}{c}{web-SF} \\
    \midrule
    Pattern $H$ & $\triangle$ & $\wedge$ & $\triangle$ & $\wedge$ & $\triangle$ & $\wedge$ &$\triangle$ & $\wedge$\\ 
    \midrule
    Time (h) & 16.7 & 15.9 & 8.2 & 7.6 & 10.6 & 9.3 & 13.5 & 12.1\\
    \bottomrule
    \end{tabular}
    }}
    \vspace{-4mm}
\end{table}

{\KXREVIEW{
\smallskip\noindent
\textbf{(5) Training. } 
%
We first report the training time of counting triangles and wedges on four datasets under massive deletion scenarios in Table~\ref{tab:train-massive}. 
We note that it normally takes several hours to train a satisfactory model. 
%
Then we study how the size of training graphs affects the model performance. We generate different graphs $G(n, p=0.5)$ by FF model with $n\in\{0.5\text{M}, 1\text{M}, 2\text{M}, 4\text{M}, 8\text{M}\}$ for training. We then generate a graph $G(n=10\text{M}, p=0.5)$ by FF model, which contains around 50M edges, for testing.
The results of counting triangles under massive deletion scenario are shown in Figure~\ref{fig:train-massive}. We observe that as the the size of graph used for training grows exponentially, the training time also increases exponentially, but the effectiveness improves slightly only. Therefore, we learn the weight function on the graphs with size around 10\%-20\% of those of the testing graphs to balance the effectiveness and the training costs. 
}}



{\KXREVIEW{
\smallskip\noindent
\textbf{(6) Relationship between an edge's weight in \texttt{WSD-L} and its associated subgraph counts.}
%
We reveal the relationship 
by collecting the statistics of (1) each edge's weight and (2) the number of triangles, which contain the edge and are formed by the end of the stream. 
Since an edge's weight depends on those edges that have been sampled in the reservoir when it appears, this value corresponds to a random variable.
Therefore, we run \texttt{WSD-L} 100 times, and for an edge, we report the mean of its weights. 
Figure~\ref{fig:relation-massive} shows the results of counting triangles on cit-PT under massive deletion scenario. 
We observe the trend that the larger an edge's weight is, the more triangles the edge is involved in, which is well aligned with the intuition behind Eq.~(\ref{eq:sv}).
}}

\begin{table}[th]
\vspace{-4mm}
\scriptsize
    \centering
    {\KXREVIEW{
    \caption{Results on transferability of \texttt{WSD-L} under massive scenario. }
    \label{tab:transfer-massive}
    \vspace{-2mm}
    \begin{tabular}{c|ccccc|c}
    \toprule
        (\textbf{Training}) & cit-HE & com-DB & soc-TX & web-SF & synthetic & \texttt{WSD-H}\\
    \midrule
        cit-PT  & \textbf{0.076} & 0.080 & \underline{0.077} & 0.078 & 0.081 & 0.083\\
        com-YT & \underline{0.049} & \textbf{0.048} & 0.053 & 0.052 & 0.050 & 0.053\\
        soc-TW & 0.653 & 0.567 & \textbf{0.451} & \underline{0.510} & 0.687 & 0.711\\
        web-GL & \underline{0.033} & 0.036 & 0.035 & \textbf{0.032} & 0.034 & 0.037\\
    \bottomrule
    \end{tabular}
    }}
    \vspace{-4mm}
\end{table}

{\KXREVIEW{
\smallskip\noindent
\textbf{(7) Transferability of \texttt{WSD-L}.} 
%
We test the transferability of \texttt{WSD-L} by applying the policy (i.e., the weight function), which has been trained on a graph under a certain category, to a graph under another category. To ensure fairness of comparison, we use the first 1M edges of the training graph in each dataset to train the model. Table~\ref{tab:transfer-massive} shows the ARE results of counting triangles. 
{\CHENG As expected, the model works the best on a graph under the same category as the training dataset and has its performance degraded to some extent when being used on a graph under a different category. In addition, the model, when used on a graph under a different category from the one of the training dataset, still works better than the heuristic-based method, which verifies the transferability of our method.}
}}

\smallskip
\noindent
\textbf{(8) Evaluation on insertion-only scenario.}
We evaluate a special case of the fully dynamic graph streams, in which it only consists of edge insertion events. Under the insertion-only scenario, \texttt{WSD-H} and \texttt{GPS-A} would be equivalent to \texttt{GPS}. 
{\CHENG{We also compare \texttt{Triest}, \texttt{ThinkD} and \texttt{WRS} under this setting.}}
The results are presented in Table~\ref{tab:res-insertion-only}. Consider the effectiveness. Two weight-sensitive sampling schemes (\texttt{WSD-L} and \texttt{GPS}) outperform other algorithms. Due to the data-driven fashion of \texttt{WSD-L}, it further improves the performance. Consider the efficiency. For \texttt{WSD-L} and \texttt{GPS}, it would cost $O(\log M)$ (worst case) to maintain the min-priority queue when an edge is inserted while for the other three algorithms, it would cost $O(1)$ to update the reservoir. Therefore, \texttt{WSD-L} and \texttt{GPS} run slightly slower.

\begin{table}[t]
\centering
\caption{Experimental results of counting triangles on cit-PT under insertion-only scenario.}
\label{tab:res-insertion-only}
\begin{tabular}{c|cccccc}
    \toprule
     & \texttt{WSD-L} & \texttt{GPS} & \texttt{Triest} & \texttt{ThinkD}  &  \texttt{WRS}  \\
    \midrule
    ARE (\%)    & \textbf{0.30} & \underline{0.34} & 0.85 & 0.41 & 0.36 \\
    MARE (\%) & \textbf{0.14} & \underline{0.20} & 0.66 & 0.24 & 0.22 \\
    Time (s) & 49.6 & 48.5 & \textbf{39.3} & \underline{40.2} & 41.1 \\
    \bottomrule
\end{tabular}
\vspace{-4mm}
\end{table}


\section{Related Work}
\label{sec:related}

\noindent
\textbf{Subgraph Counting}. 
Subgraph counting (e.g., triangle counting) in a streaming graph is a fundamental problem in graph analysis and has been extensively studied.
For example, many studies \cite{stefani2017triest, wang2017approximately, ahmed2017sampling, jung2019furl, shin2018think, shin2020fast, zhang2020reservoir, lee2020temporal, tsourakakis2009doulion, lim2015mascot, wang2019rept, etemadi2019pes, lim2018memory, yu2019distributed, han2017counting} propose algorithms to estimate triangle counts in a streaming graph. Among these studies, some studies~\cite{stefani2017triest, wang2017approximately, ahmed2017sampling, jung2019furl, shin2018think, shin2020fast, zhang2020reservoir,lee2020temporal} assume some space constraint for sampled edges while others \cite{tsourakakis2009doulion, lim2015mascot, wang2019rept, etemadi2019pes, lim2018memory, yu2019distributed, han2017counting} assume no constraint on the number of sampled edges.
Our work follows the settings in \cite{stefani2017triest, lee2020temporal, shin2020fast, shin2017wrs, shin2018think}, which target fully dynamic graph steams and assume a constraint on the number of sampled edges.
{\KXREVIEW{
These existing methods all adopt the \emph{random pairing} technique \cite{gemulla2006dip}, which extends the standard reservoir sampling method to support deletions, where each deletion from the stream is {\CHENGREVIEW ``paired with'' (or compensated by)} a subsequent insertion. Specifically, \texttt{Triest} \cite{stefani2017triest} is the first method to estimate the triangle count in fully graph streams with random pairing, where the estimation is only updated when an edge is sampled. 
Later, \cite{shin2020fast, shin2018think} extend \cite{stefani2017triest} and propose \texttt{ThinkD}. Instead of simply discarding the unsampled edges as \texttt{Triest} {\CHENGREVIEW does}, \texttt{ThinkD} first updates the estimation and then updates the sampled graph, which results in a smaller variance. 
More recently, \texttt{WRS} \cite{lee2020temporal, shin2017wrs} is proposed to further exploit the temporal locality information. It divides the storage budget into two parts, namely waiting room and reservoir, and stores the most recent edges in the waiting room unconditionally and the sampled edges in the reservoir.
}}
As explained in Section~\ref{sec:intro}, these existing studies~\cite{stefani2017triest, shin2020fast, lee2020temporal, shin2018think, shin2017wrs} all sample the edges with uniform probabilities, e.g., each edge is treated equally for sampling, which would likely result in sub-optimal samples.
%
%
In addition, there are some other studies on the problem of counting subgraphs in a static graph \cite{alon1997finding, arifuzzaman2017distributed, arifuzzaman2019fast, ouyang2019memory,dave2017clog, wang2014efficiently,han2016waddling, yang2018ssrw, chen2016general, saha2015finding, huangtowards, ahmed2016estimation, bressan2017counting, bressan2018motif, bressan2019motivo, klusowski2018counting}, which is different from the setting targeted in our work.
{\KXREVIEW{
It is worth noting that there are several recent studies, which develop  learning-based algorithms for subgraph counting problem
\cite{chen2020can, liu2020neural, zhao2021learned, wang2022neural, roy2022interpretable}. They differ from our study in (1) all these studies target static graphs; (2) \cite{zhao2021learned, wang2022neural, liu2020neural} target the problem of counting subgraph isomorphism matches, where each vertex has a label.}}
%
Some other studies~\cite{wu2017counting,wu2018counting} target the subgraph counting problem on online social networks. Interested readers are referred to two surveys \cite{al2018triangle, ribeiro2021survey} for more details of subgraph counting algorithms.

\smallskip
\noindent
\textbf{Reservoir Sampling}. Reservoir sampling was first studied in \cite{vitter1985random}, which samples items in a stream with equal probabilities. \cite{nath2004synopsis} implements the simple reservoir sampling via min-wise sampling, which provides the same effect as \cite{vitter1985random} but consumes more memory. To handle fully dynamic streams with deletions, \cite{gemulla2006dip} proposes a framework, namely random pairing, where each deletion from the stream is mapped to a subsequent insertion. Random pairing guarantees uniformity, and thus it is widely used for sampling fully dynamic streams in which items have the same weight. 
To sample weighted items, \cite{duffield2007priority, efraimidis2006weighted} both propose a rank-based sampling scheme. They differ in their priority functions and are used in different applications. However, these two methods can only handle insertion-only streams while our weighted sampling method can handle fully dynamic streams.

\smallskip
\noindent
\textbf{Reinforcement Learning}. Given a specific environment, which is generally formulated as a Markov Decision Process (MDP) \cite{puterman2014markov}, reinforcement learning helps an agent in the environment learn how to map the situations to actions so as to maximize the accumulative rewards \cite{sutton2018reinforcement}. In this paper, we model 
the process of deciding the weights of edges in a stream sequentially as an MDP and use a popular policy gradient method DDPG \cite{lillicrap2016continuous} for solving the problem. To the best of knowledge, this is the first deep reinforcement learning based solution applied to the subgraph counting in graph streams.


\section{Conclusion}
\label{sec:conclusion}

In this paper, we study the subgraph counting problem in fully dynamic graph streams. 
We propose a weighted sampling algorithm \texttt{WSD}, which {\chengfinal samples} edges non-uniformly based on their weights, and construct an unbiased estimator based on the sampled edges by \texttt{WSD}.
Furthermore, we develop a reinforcement learning based method for setting weights of edges in a data-driven fashion.
Compared with existing algorithms, our algorithms can produce {\chengfinal estimations} with smaller errors and often run faster. 
One interesting direction for future research is to 
extend the reinforcement learning enhanced \texttt{WSD} method to other problems on fully dynamic graphs.

\smallskip
\noindent\textbf{Acknowledgments.}
This research is supported by the Ministry of Education, Singapore, under its Academic Research Fund (Tier 2 Award MOE-T2EP20221-0013 and Tier 1 Award (RG77/21)). Any opinions, findings and conclusions or recommendations expressed in this material are those of the author(s) and do not reflect the views of the Ministry of Education, Singapore.
This work was also partially supported by the A*STAR Cyber-Physical Production
System (CPPS) – Towards Contextual and Intelligent Response Research Program, under the RIE2020 IAF-PP Grant A19C1a0018, and Model Factory@SIMTech.

\bibliographystyle{IEEEtran}
\bibliography{reference}

\clearpage


\appendix

\subsection{Proof of Theorem~\ref{theo:gps-a}}
\label{app:proof-gps-a}

\begin{proof}
Since the sampling process of \texttt{GPS-A} is exactly the same as that of \texttt{GPS} and we simply attach a ``DEL'' tag to the corresponding edge when a deletion event happens, the probability that a set of edges is included in the reservoir would be {\CHENG equal} to that of \texttt{GPS}, i.e., Eq.~(\ref{eq:gps-edges}). 
%
{\CHENG We then prove the following two statements.}
\begin{equation}
\label{eq:gps-a-xy}
    \mathbb{E}[X_{\texttt{GPS-A}}^{J}] = 1, \quad \mathbb{E}[Y_{\texttt{GPS-A}}^{J}] = 1. 
\end{equation}
Note that $X_{\texttt{GPS-A}}^J$ is positive only when all edges $e\in J\setminus e_{i_{|H|}}$ are in the reservoir $\mathcal{R}$ at the end of $t_a(J)-1$. Then at the beginning of $t_a(J)$, by applying Eq.~(\ref{eq:gps-edges}), we have
\begin{equation}
    \mathbb{P}[(J \setminus e_{i_{|H|}}) \subset \mathcal{R} ]
    = \prod_{e \in J\setminus e_{i_{|H|}}} \mathbb{P}[r(e)>r_{M+1}].
\end{equation}
Thus, $\mathbb{E}[X_{\texttt{GPS-A}}^{J}] = 1$ holds. The other statement can be verified similarly. Finally, based on the linearity of expectation and Eq.~(\ref{eq:gps-a-xy}), we have the following deductions.
\begin{equation}
\begin{aligned}
    & \mathbb{E}[c_{\texttt{GPS-A}}^{(t)}] = \sum_{J\in \mathcal{A}^{(t)}} \mathbb{E}[X_{\texttt{GPS-A}}^{J}] - \sum_{J\in \mathcal{D}^{(t)}} \mathbb{E}[Y_{\texttt{GPS-A}}^{J}] \\
    &= \sum_{J \in \mathcal{A}^{(t)}} 1 - \sum_{J \in \mathcal{D}^{(t)}} 1 
    = |\mathcal{A}^{(t)}| - |\mathcal{D}^{(t)}| =  |\mathcal{J}^{(t)}|, 
\end{aligned}
\end{equation}
which completes the proof. 
\end{proof}




\subsection{Proof of Theorem~\ref{theo:wsd}}
\label{app:proof-wsd}

\begin{proof}

We first prove by deduction that at the end of time $t$, the probability that a set of edges $E=\{e_1, \cdots, e_{|E|}\}$ $(|E| \le M)$ is included in the reservoir $\mathcal{R}^{(t)}$ is as follows.
\begin{equation}
\label{eq:prob-set}
    \mathbb{P}[E\subset \mathcal{R}^{(t)}] = \displaystyle\prod\limits_{e\in E} \mathbb{P}[e\in \mathcal{R}^{(t)}] = \displaystyle\prod\limits_{e\in E} \mathbb{P}[r(e)>\tau_q], 
\end{equation}
where $\tau_q$ is observed at time $t$. 
%
When $t \le M$, since (1) the set of the edges would be in the reservoir for sure, and (2) $\tau_q$ is always 0, we know that Eq.~(\ref{eq:prob-set}) holds. 
Assume that Eq.~(\ref{eq:prob-set}) holds for $t = k$ $(k \ge M)$.
Consider $t = k + 1$. We have three cases (as defined before).
\begin{itemize}
    \item \textbf{Case 1 and Case 3.}
    Eq.~(\ref{eq:prob-set}) holds since the probabilities that edges in $E$ are included are not changed. 
    
    \item \textbf{Case 2.} 
    $E$ would be in the reservoir $\mathcal{R}^{(k+1)}$ only if all edges in $E$ are not dropped after the insertion, the probability of which is 
    \begin{equation}
    \label{eq:prob-set-cond}
        \displaystyle\prod\limits_{e\in E} \mathbb{P}[e\in \mathcal{R}^{(k+1)} \mid e\in \mathcal{R}^{(k)}]
        = \displaystyle\prod\limits_{e\in E} \frac{\mathbb{P}[e\in \mathcal{R}^{(k+1)}]}{\mathbb{P}[e\in \mathcal{R}^{(k)}]}
    \end{equation}
    Then, we deduce that the probability that $E$ is in the reservoir $\mathcal{R}^{(k+1)}$ is as follows.
    \begin{equation}
    \begin{aligned}
        & \mathbb{P}[E\subset \mathcal{R}^{(k+1)}] \\ 
        &= \mathbb{P}[E\subset \mathcal{R}^{(k)}] \cdot \mathbb{P}[E\subset \mathcal{R}^{(k+1)} \mid E \subset \mathcal{R}^{(k)}] \\
        &= \displaystyle\prod\limits_{e\in E} \mathbb{P}[e\in \mathcal{R}^{(k)}] \cdot \displaystyle\prod\limits_{e\in E} \frac{\mathbb{P}[e\in \mathcal{R}^{(k+1)}]}{\mathbb{P}[e\in \mathcal{R}^{(k)}]} \\
        &= \displaystyle\prod\limits_{e\in E} \mathbb{P}[e\in \mathcal{R}^{(k+1)}] = \displaystyle\prod\limits_{e\in E} \mathbb{P}[r(e)>\tau_q^{(k+1)}]
    \end{aligned}
    \end{equation}
    
\end{itemize}

{\KAIXIN
So far we have assumed that $e_{k+1} \notin E$.  
Next, we consider the case that $e_{k+1}$ is in $E$. In this case, we define $E'=E\setminus e_{k+1}$. 
we have the following deductions.
\begin{equation}
\label{eq:17}
\begin{aligned}
    &\mathbb{P}[E\subset \mathcal{R}^{(k+1)}] = \mathbb{P}[(E'\cup e_{k+1})\subset \mathcal{R}^{(k+1)}] \\
    &= \mathbb{P}[(E' \subset \mathcal{R}^{(k+1)} \land e_{k+1} \in \mathcal{R}^{(k+1)}] \\
    &= \mathbb{P}[e_{k+1} \in \mathcal{R}^{(k+1)}] \cdot \mathbb{P}[(E' \subset \mathcal{R}^{(k+1)} \mid e_{k+1} \in \mathcal{R}^{(k+1)}] \\
    &=\mathbb{P}[e_{k+1}\in \mathcal{R}^{(k+1)}] \cdot \mathbb{P}[E'\subset \mathcal{R}^{(k)}] \\
    & \quad\quad\quad\quad  \cdot \mathbb{P}[E'\subset \mathcal{R}^{(k+1)} \mid E' \subset \mathcal{R}^{(k)}, e_{k+1} \in \mathcal{R}^{(k+1)}] \\
    &= \mathbb{P}[r(e_{k+1})> \tau_q^{(k+1)}] \cdot \displaystyle\prod\limits_{e\in E'} \mathbb{P}[r(e)>\tau_q^{(k)}] \\
    & \quad\quad\quad\quad  \cdot \mathbb{P}[E'\subset \mathcal{R}^{(k+1)} \mid E' \subset \mathcal{R}^{(k)}, e_{k+1} \in \mathcal{R}^{(k+1)}] \\
\end{aligned}
\end{equation}

The first equation is derived from the definition of $E$. The third and forth equations are based on conditional probability. The fifth equation applies Lemma~\ref{lemma:prob} and the assumption that Eq.~(\ref{eq:prob-set}) holds when $t=k$ to the forth equation. 
We consider the conditional probability $\mathbb{P}[E'\subset \mathcal{R}^{(k+1)} \mid E' \subset \mathcal{R}^{(k)}, e_{k+1} \in \mathcal{R}^{(k+1)}]$ in the following cases.

For Case 1, we have $\tau_q^{(k+1)} = \tau_q^{(k)}$. Since the reservoir is not full, the conditional probability would be equal to 1. Thus, Eq.~(\ref{eq:prob-set}) holds when $t =k+1$. 

For Case 2, 
since the reservoir is full before the insertion of $e_{k+1}$, then $e_{k+1}\in \mathcal{R}^{(k+1)}$ means that the edge with the minimum rank in $\mathcal{R}^{(k)}$ is dropped, and $\tau_q^{(k+1)}$ is updated. Thus, the conditional probability is equal to 
\begin{equation}
\begin{aligned}
\label{eq:18}
    & \mathbb{P}[E'\subset \mathcal{R}^{(k+1)} \mid E' \subset \mathcal{R}^{(k)}, e_{k+1} \in \mathcal{R}^{(k+1)}] \\
    &= \displaystyle\prod\limits_{e\in E'} \frac{\mathbb{P}[e\in \mathcal{R}^{(k+1)}]}{\mathbb{P}[e\in \mathcal{R}^{(k)}]} = \displaystyle\prod\limits_{e\in E'} \frac{\mathbb{P}[r(e) > \tau_q^{(k+1)}]}{\mathbb{P}[r(e) > \tau_q^{(k)}]} \\
\end{aligned}
\end{equation}
By applying Eq.~(\ref{eq:18}) to Eq.~(\ref{eq:17}), we know Eq.~(\ref{eq:prob-set}) holds when $t=k+1$. Therefore, Eq.~(\ref{eq:prob-set}) holds $\forall t$. 
}

\smallskip
Next, we prove the following two statements. 
\begin{equation}
\label{eq:x&y}
    \mathbb{E}[X_{\texttt{WSD}}^{J}] = 1, \quad \mathbb{E}[Y_{\texttt{WSD}}^{J}] = 1. 
\end{equation}
Note that $X_{\texttt{WSD}}^J$ is positive only when all edges $e\in J\setminus e_{i_{|H|}}$ are in the reservoir $\mathcal{R}$ at the end of $t_a(J)-1$. Then at the beginning of $t_a(J)$, 
by applying Eq.~(\ref{eq:prob-set}), we have
\begin{equation}
    \mathbb{P}[(J \setminus e_{i_{|H|}}) \subset \mathcal{R} ]
    = \prod_{e \in J\setminus e_{i_{|H|}}} \mathbb{P}[r(e)>\tau_q].
\end{equation}
Thus, $\mathbb{E}[X_{\texttt{WSD}}^{J}] = 1$ holds. The other statement can be verified similarly.

Finally, based on the linearity of expectation and Eq.~(\ref{eq:x&y}), we have the following deductions.
\begin{equation}
\begin{aligned}
    & \mathbb{E}[c_{\texttt{WSD}}^{(t)}] = \sum_{J\in \mathcal{A}^{(t)}} \mathbb{E}[X_{\texttt{WSD}}^{J}] - \sum_{J\in \mathcal{D}^{(t)}} \mathbb{E}[Y_{\texttt{WSD}}^{J}] \\
    &= \sum_{J \in \mathcal{A}^{(t)}} 1 - \sum_{J \in \mathcal{D}^{(t)}} 1 
    = |\mathcal{A}^{(t)}| - |\mathcal{D}^{(t)}| =  |\mathcal{J}^{(t)}|, 
\end{aligned}
\end{equation}
which completes the proof. 
\end{proof}

\subsection{Additional Experimental Results}
\label{app:exp}

\begin{table}[t]
\scriptsize
\centering
\caption{Results of counting 4-cliques under the massive deletion scenario, where we use the default values of the parameters, $M=200,000$, $\alpha=3,000,000^{-1}$ and $\beta_m=0.8$.}
\vspace{-2mm}
\label{tab:res-4clique-massive}
\begin{tabular}{r|cccccc}
    \toprule
     \textbf{Graph} & \texttt{WSD-L} & \texttt{WSD-H} & \texttt{GPS-A} & \texttt{Triest} & \texttt{ThinkD} & \texttt{WRS} \\ 
    \midrule
    & \multicolumn{6}{c}{Absolute Relative Error ($\%$)} \\ \cmidrule{2-7}
    cit-PT & \textbf{0.771} &	\underline{0.880} & 0.962 &	1.365 &	1.114 &	0.947\\
    com-YT  &  \textbf{0.481} &	\underline{0.551} & 0.684 &   1.330 &	1.046 &	0.822\\
    web-GL &  \textbf{0.582} &	\underline{0.666} & 0.747 &	1.229 &	1.099 &	0.847\\
    synthetic &  \textbf{2.843} & \underline{3.207} & 3.582 & 3.913 & 3.764 & 3.368 \\
    \midrule
    & \multicolumn{6}{c}{Mean Absolute Relative Error ($\%$)} \\ \cmidrule{2-7}
    cit-PT  & \textbf{0.811} &	\underline{0.915} & 0.941 &	1.361 &	1.040 &	0.922\\
    com-YT  & \textbf{0.833} &	\underline{0.894} & 0.976 &	1.858 &	1.407 &	1.027\\
    web-GL & \textbf{0.391} &	\underline{0.466} & 0.644 &	1.093 &	0.893 &	0.729\\
    synthetic  & \textbf{3.064} & \underline{3.291} & 3.563 & 4.470 & 4.115 & 3.764 \\
    \midrule
    & \multicolumn{6}{c}{Running Time (s)} \\ \cmidrule{2-7}
    cit-PT & \underline{73.5} &	\textbf{69.1} & 76.3 &	270.8	& 273.5 &	277.4\\
    com-YT & 16.8 &	16.3 & 19.2 &	\textbf{14.8} &	\underline{15.1} &	15.7  \\
    web-GL & \underline{25.9} &	\textbf{21.1}  & 30.4 &	65.1 &	67.7 &	70.6\\
    synthetic & \underline{64.3} & \textbf{61.5} & 67.4 & 328.4 & 339.5 & 351.5 \\
    \bottomrule
\end{tabular}
\vspace{-2mm}
\end{table}

\begin{table}[th]
\scriptsize
\centering
\caption{Results of counting wedges under the light deletion scenario, where we use the default values of the parameters, $M=200,000$ and $\beta_l=0.2$.}
\vspace{-2mm}
\label{tab:res-wedge-light}
\begin{tabular}{c|cccccc}
    \toprule
    \textbf{Graph} & \texttt{WSD-L} & \texttt{WSD-H} &\texttt{GPS-A} & \texttt{Triest} & \texttt{ThinkD}  &  \texttt{WRS}  \\ 
    \midrule
    & \multicolumn{6}{c}{Absolute Relative Error ($\%$)} \\ \cmidrule{2-7}
    cit-PT & \textbf{0.009} & \underline{0.010} & 0.025 & 0.062 & 0.053 & 0.035 \\
    com-YT  & \textbf{0.006} & \underline{0.008} & 0.058 & 0.289 & 0.277 & 0.158 \\
    {\KXREVIEW{soc-TW}} & {\KXREVIEW{\textbf{0.343}}} & {\KXREVIEW{\underline{0.421}}} & {\KXREVIEW{0.509}} & {\KXREVIEW{0.657}} & {\KXREVIEW{0.654}} & {\KXREVIEW{0.603}} \\
    web-GL & \textbf{0.042} & \underline{0.046} &  0.077& 0.429 & 0.347 & 0.128 \\
    synthetic  & \textbf{0.014} & \underline{0.021} & 0.028 & 0.103 & 0.038 & 0.022 \\
    \midrule
    & \multicolumn{6}{c}{Mean Absolute Relative Error ($\%$)} \\ \cmidrule{2-7}
    cit-PT & \textbf{0.007} & \underline{0.008} & 0.024 & 0.057 & 0.046 & 0.033 \\
    com-YT  & \textbf{0.005} & \underline{0.006} & 0.043 & 0.101 & 0.097 & 0.053 \\
    {\KXREVIEW{soc-TW}} & {\KXREVIEW{\textbf{0.391}}} & {\KXREVIEW{\underline{0.482}}} & {\KXREVIEW{0.744}} & {\KXREVIEW{1.139}} & {\KXREVIEW{1.057}} & {\KXREVIEW{0.950}} \\
    web-GL & \textbf{0.014} & \underline{0.026} & 0.053 & 0.344 & 0.211 & 0.075 \\
    synthetic & \textbf{0.017} & \underline{0.030} & 0.042 & 0.111 & 0.054 & 0.034 \\
    \midrule
    & \multicolumn{6}{c}{Running Time (s)} \\ \cmidrule{2-7}
    cit-PT & \underline{57.8} & \textbf{53.6} & 60.7 & 229.4 & 235.4 & 233.3 \\
    com-YT & \underline{14.2} & \textbf{13.1} & 15.9 & 41.6 & 42.5 & 42.7 \\
    {\KXREVIEW{soc-TW}} & {\KXREVIEW{\underline{1025.5}}} & {\KXREVIEW{\textbf{1014.2}}} & {\KXREVIEW{2078.3}} & {\KXREVIEW{3367.3}} & {\KXREVIEW{3389.1}} & {\KXREVIEW{3466.8}} \\
    web-GL & \underline{24.7} & \textbf{22.1} & 25.6 & 79.8 & 81.6 & 82.6 \\
    synthetic & \underline{50.0} & \textbf{49.1} & 53.2 & 133.1 & 133.3 & 135.2 \\
    \bottomrule
\end{tabular}
\vspace{-4mm}
\end{table}

\begin{table}[th]
\scriptsize
\centering
\caption{Results of counting triangles under the light deletion scenario, where we use the default values of the parameters, $M=200,000$ and $\beta_l=0.2$.}
\vspace{-2mm}
\label{tab:res-triangle-light}
\begin{tabular}{c|cccccc}
    \toprule
    \textbf{Graph} & \texttt{WSD-L} & \texttt{WSD-H} & \texttt{GPS-A} & \texttt{Triest} & \texttt{ThinkD}  &  \texttt{WRS}  \\ 
    \midrule
    & \multicolumn{6}{c}{Absolute Relative Error ($\%$)} \\ \cmidrule{2-7}
    cit-PT & \textbf{0.171} & \underline{0.221} & 0.257 & 0.834 & 0.293 & 0.224 \\
    com-YT & \textbf{0.051} & \underline{0.059} & 0.104 & 0.941 & 0.797 & 0.471 \\
    {\KXREVIEW{soc-TW}} & {\KXREVIEW{\textbf{0.564}}} & {\KXREVIEW{\underline{0.762}}} & {\KXREVIEW{1.109}} & {\KXREVIEW{1.484}} & {\KXREVIEW{1.333}} & {\KXREVIEW{1.279}} \\
    web-GL & \textbf{0.061} & \underline{0.069} & 0.153 & 0.591 & 0.270 & 0.301 \\
    synthetic & \textbf{0.049} & \underline{0.067} & 0.114 & 0.652 & 0.441 & 0.233 \\
    \midrule
    & \multicolumn{6}{c}{Mean Absolute Relative Error ($\%$)} \\ \cmidrule{2-7}
    cit-PT & \textbf{0.183} & \underline{0.236} & 0.286 & 1.048 & 0.336 & 0.254 \\
    com-YT & \textbf{0.049} & \underline{0.053} & 0.133 & 0.552 & 0.475 & 0.367 \\
    {\KXREVIEW{soc-TW}} & {\KXREVIEW{\textbf{0.651}}} & {\KXREVIEW{\underline{0.793}}} & {\KXREVIEW{1.164}} & {\KXREVIEW{1.748}} & {\KXREVIEW{1.523}} & {\KXREVIEW{1.404}}\\
    web-GL & \textbf{0.039} & \underline{0.049} & 0.153 & 0.328 & 0.195 & 0.266 \\
    synthetic & \textbf{0.040} & \underline{0.053} &  0.082 & 0.361 & 0.182 & 0.122\\
    \midrule
    & \multicolumn{6}{c}{Running Time (s)} \\ \cmidrule{2-7}
    cit-PT & \underline{61.5} & \textbf{58.0} & 63.7 &  236.1 & 238.9 & 240.4 \\
    com-YT & \underline{17.0} & \textbf{16.1} & 19.1 & 54.7 & 55.5 & 56.2 \\
    {\KXREVIEW{soc-TW}} & {\KXREVIEW{\underline{1314.5}}} & {\KXREVIEW{\textbf{1267.4}}} & {\KXREVIEW{2833.9}} & {\KXREVIEW{4317.8}} & {\KXREVIEW{4320.3}} & {\KXREVIEW{4411.7}} \\
    web-GL & \underline{21.1} & \textbf{19.5} & 23.3 & 68.8 & 70.3 & 72.8 \\
    synthetic & \underline{36.1} & \textbf{35.6} & 38.2 & 134.7 & 137.6 & 139.3\\
    \bottomrule
\end{tabular}
\vspace{-2mm}
\end{table}

\begin{table}[th]
\scriptsize
\centering
\caption{Results of counting 4-cliques under the light deletion scenario, where we use the default values of the parameters, $M=200,000$ and $\beta_l=0.2$.}
\vspace{-2mm}
\label{tab:res-4clique-light}
\begin{tabular}{c|cccccc}
    \toprule
    \textbf{Graph} & \texttt{WSD-L} & \texttt{WSD-H} & \texttt{GPS-A} & \texttt{Triest} & \texttt{ThinkD}  &  \texttt{WRS}  \\ 
    \midrule
    & \multicolumn{6}{c}{Absolute Relative Error ($\%$)} \\ \cmidrule{2-7}
    cit-PT & \textbf{1.156} &	\underline{1.320} & 1.572 &	2.593 &	1.782 &	1.420 \\
    com-YT & \textbf{1.300} &	\underline{1.563} & 1.728 &	2.856 &	2.653 &	2.295 \\
    web-GL & \textbf{0.814} &	\underline{1.198} & 1.302 &	1.966 &	1.538 &	1.439 \\
    synthetic &  \textbf{0.834} & \underline{0.891} & 1.043 &  1.419 & 1.247 & 1.162 \\ 
    \midrule
    & \multicolumn{6}{c}{Mean Absolute Relative Error ($\%$)} \\ \cmidrule{2-7}
    cit-PT & \textbf{1.135} &	\underline{1.372} & 1.521 &	2.177 &	1.768 &	1.475 \\
    com-YT & \textbf{1.166} &	\underline{1.788} & 2.439 &	2.675 &	2.892 &	2.854 \\
    web-GL & \textbf{0.742} &	\underline{0.932} & 1.124 &	2.186 &	1.607 &	1.020 \\
    synthetic & \textbf{1.154} & \underline{1.243} & 1.388 & 1.761 & 1.679 & 1.564 \\
    \midrule
    & \multicolumn{6}{c}{Running Time (s)} \\ \cmidrule{2-7}
    cit-PT & {\KXREVIEW{\underline{76.0}}} &    {\KXREVIEW{\textbf{73.5}}} & 80.4 &	224.9 &	232.5 &	273.8 \\
    com-YT & {\KXREVIEW{\underline{17.8}}} & {\KXREVIEW{\textbf{16.6}}} & 19.2 &	217.2 &	230.1 &	251.7 \\
    web-GL & {\KXREVIEW{\underline{31.0}}} & {\KXREVIEW{\textbf{25.3}}}	 & 35.9 &	78.1 &    81.24 &	84.7  \\
    synthetic & \underline{60.4} & \textbf{58.9} & 63.4 & 180.3 & 192.1 & 197.7  \\
    \bottomrule
\end{tabular}
\vspace{-4mm}
\end{table}

\smallskip\noindent
\textbf{(1) Evaluation on real and synthetic datasets.}
The results of counting 4-cliques under massive deletion scenario are shown in Table~\ref{tab:res-4clique-massive}. The results of counting wedges, triangles and 4-cliques under light deletion scenario are shown in Table~\ref{tab:res-wedge-light}, Table~\ref{tab:res-triangle-light} and Table~\ref{tab:res-4clique-light}, respectively.

\begin{figure}[th]
\vspace{-6mm}
    \centering
	\subfigure[ARE (light deletion scenario)]{
	    \label{fig:scale-are-light}
		\includegraphics[width=0.23\textwidth]{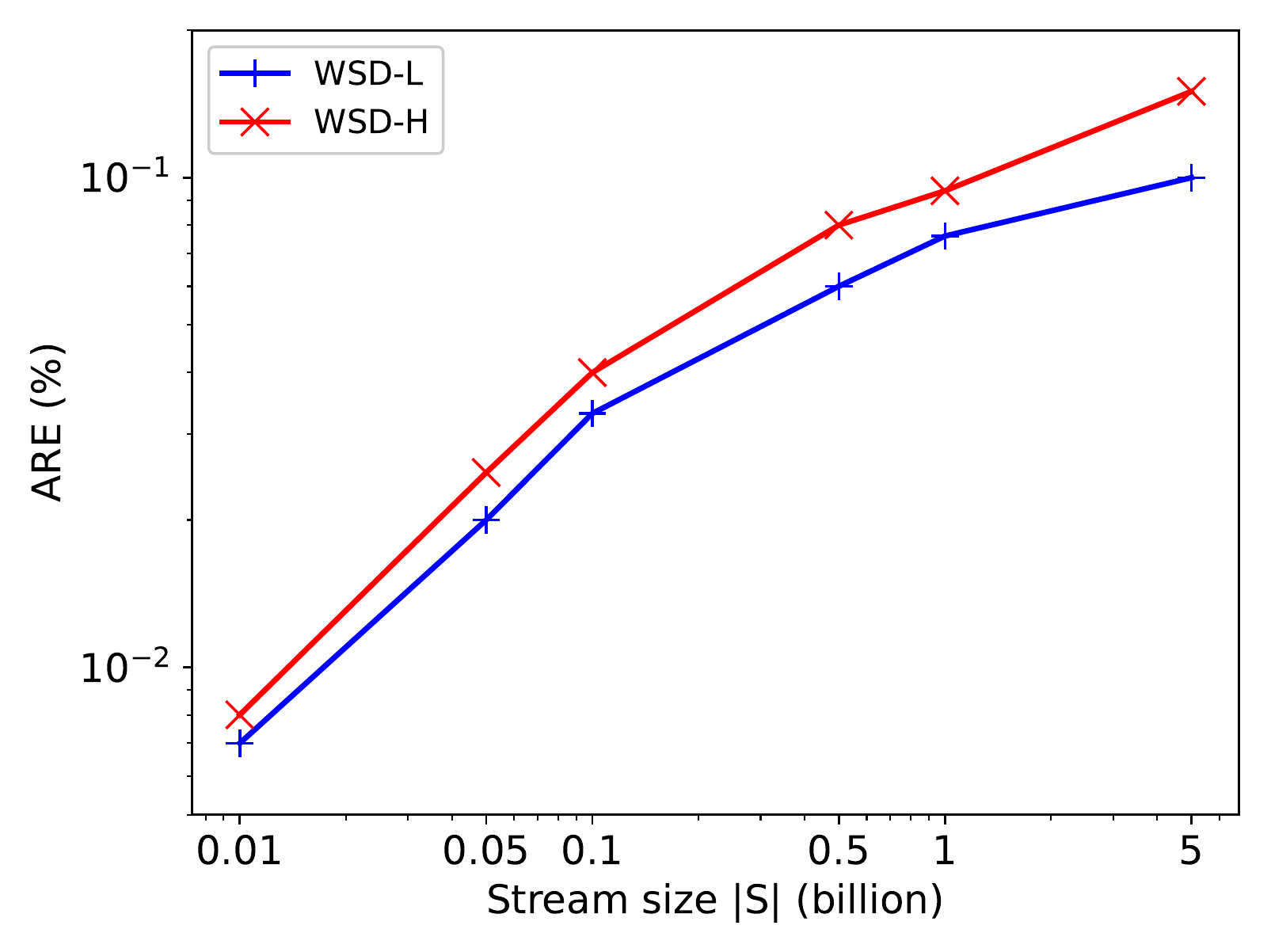}}
	\subfigure[Time (light deletion scenario)]{
	    \label{fig:scale-time-light}
		\includegraphics[width=0.23\textwidth]{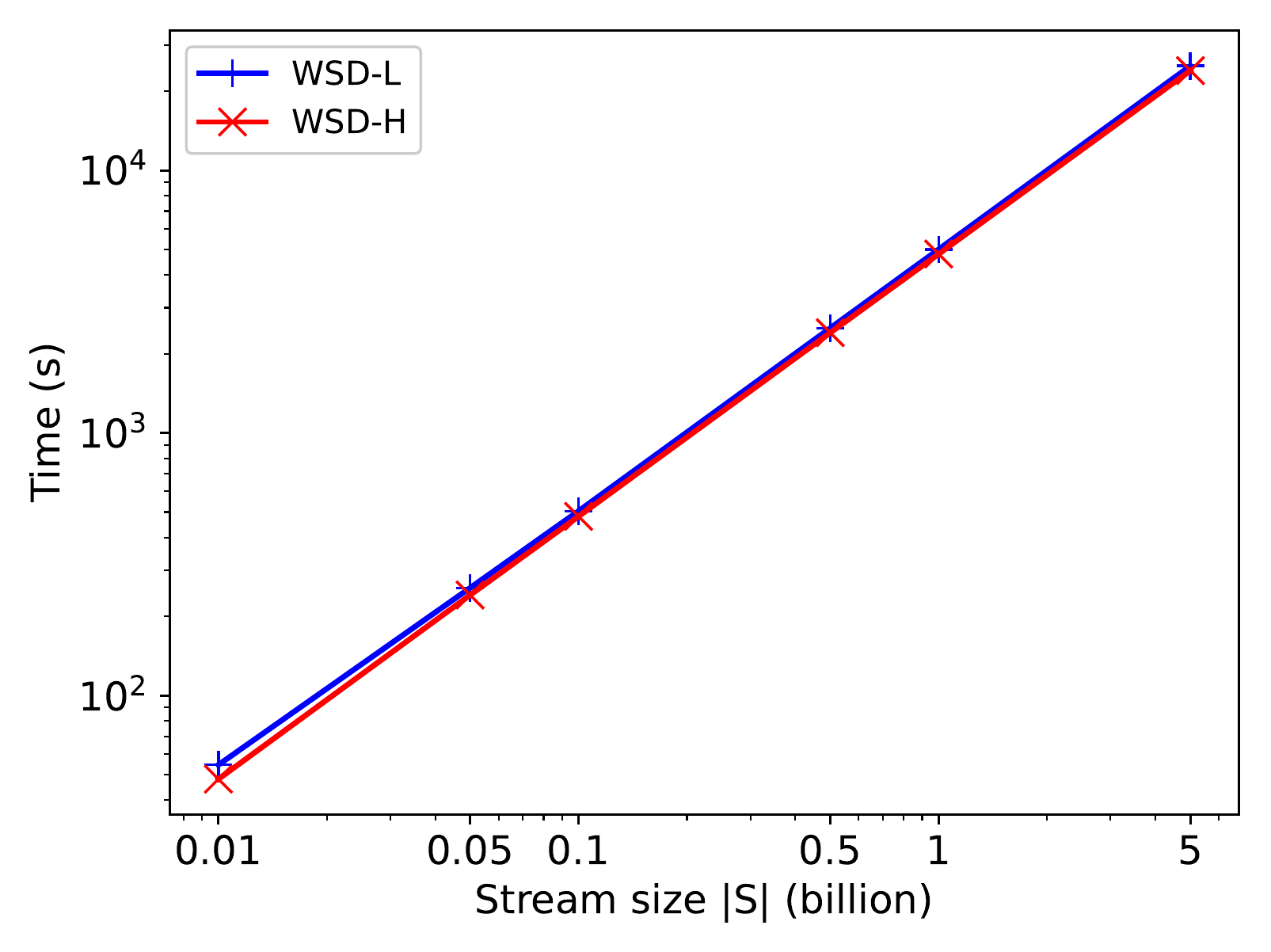}}
    \vspace{-2mm}
    \caption{Results on scalability test of counting triangles, reporting the ARE and running time of \texttt{WSD-L} and \texttt{WSD-H} on different size of streams (synthetic). }
    \vspace{-4mm}
    \label{fig:scale-light}
\end{figure}

\smallskip\noindent
\textbf{(2) Scalability.}
The results of the scalability test of \texttt{WSD} under light deletion scenario are shown in Figure~\ref{fig:scale-light}.

\smallskip\noindent
\textbf{(3) Impact of the ordering of the stream.}
Figure~\ref{fig:light-are-o} shows the ARE results of counting triangles on cit-PT under light deletion scenario.

\smallskip\noindent
\textbf{(4) Effects of $M$}. Figure~\ref{fig:light-are-m} shows the results of counting triangles on cit-PT under the light deletion scenario.

\begin{figure*}[th]
    \centering
    \subfigure[Ordering of the stream (cit-PT)]{
	    \label{fig:light-are-o}
		\includegraphics[width=0.235\textwidth]{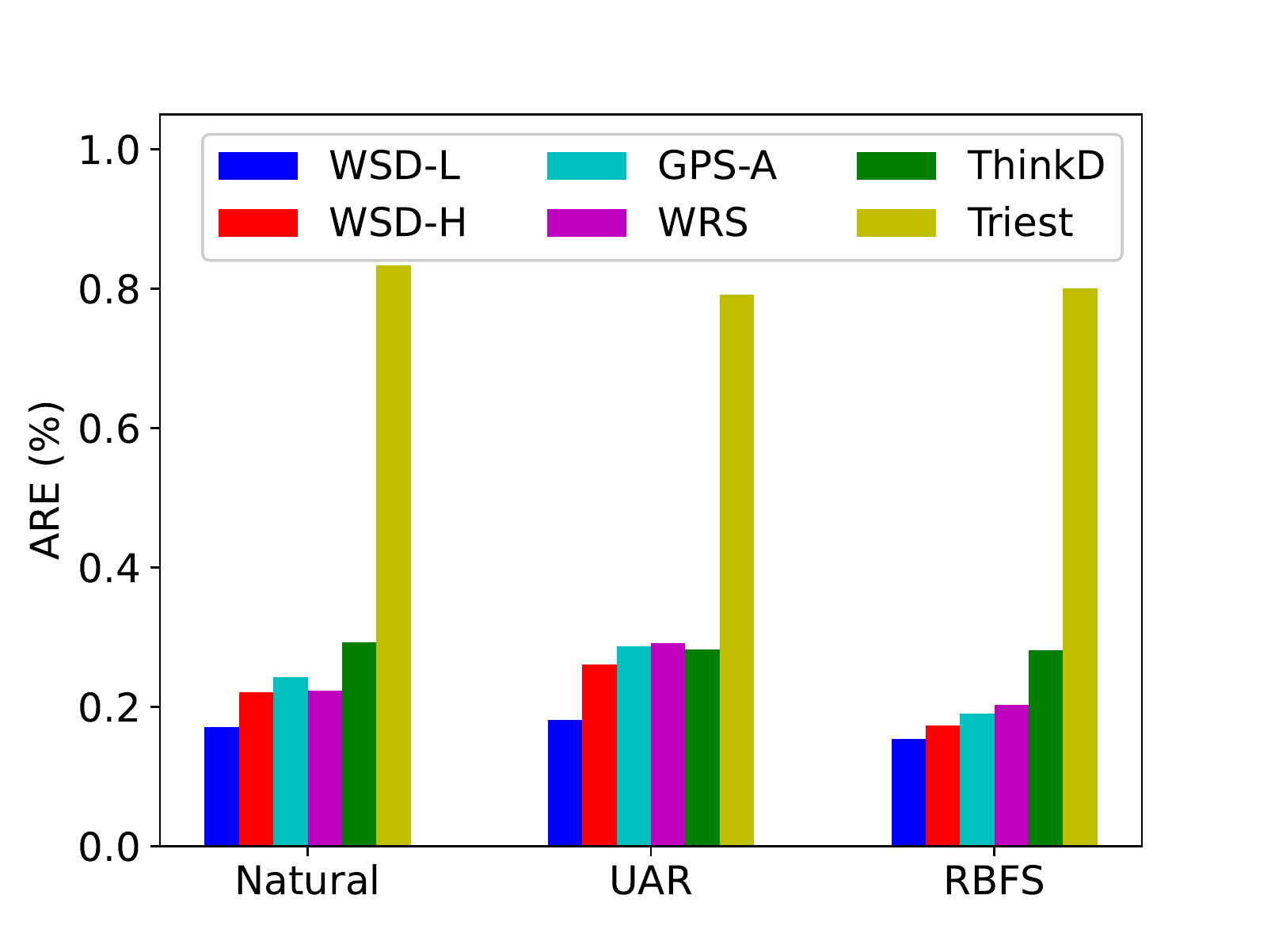}}
	\subfigure[Max. reservoir size $M$ (cit-PT)]{
	    \label{fig:light-are-m}
		\includegraphics[width=0.235\textwidth]{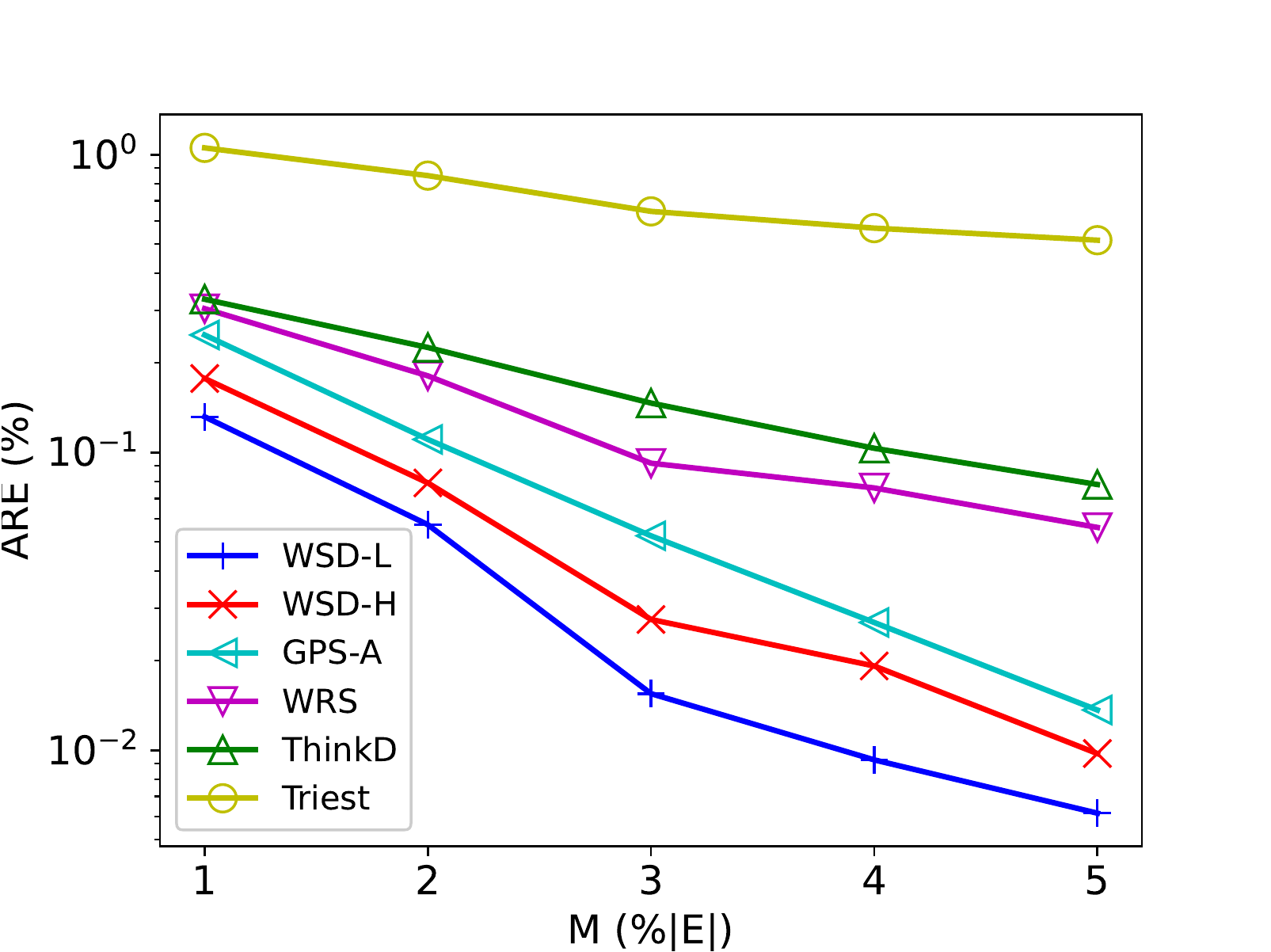}}
    \subfigure[{\KXREVIEW{Training size} (synthetic)}]{
	    \label{fig:train-light}
		\includegraphics[width=0.22\textwidth]{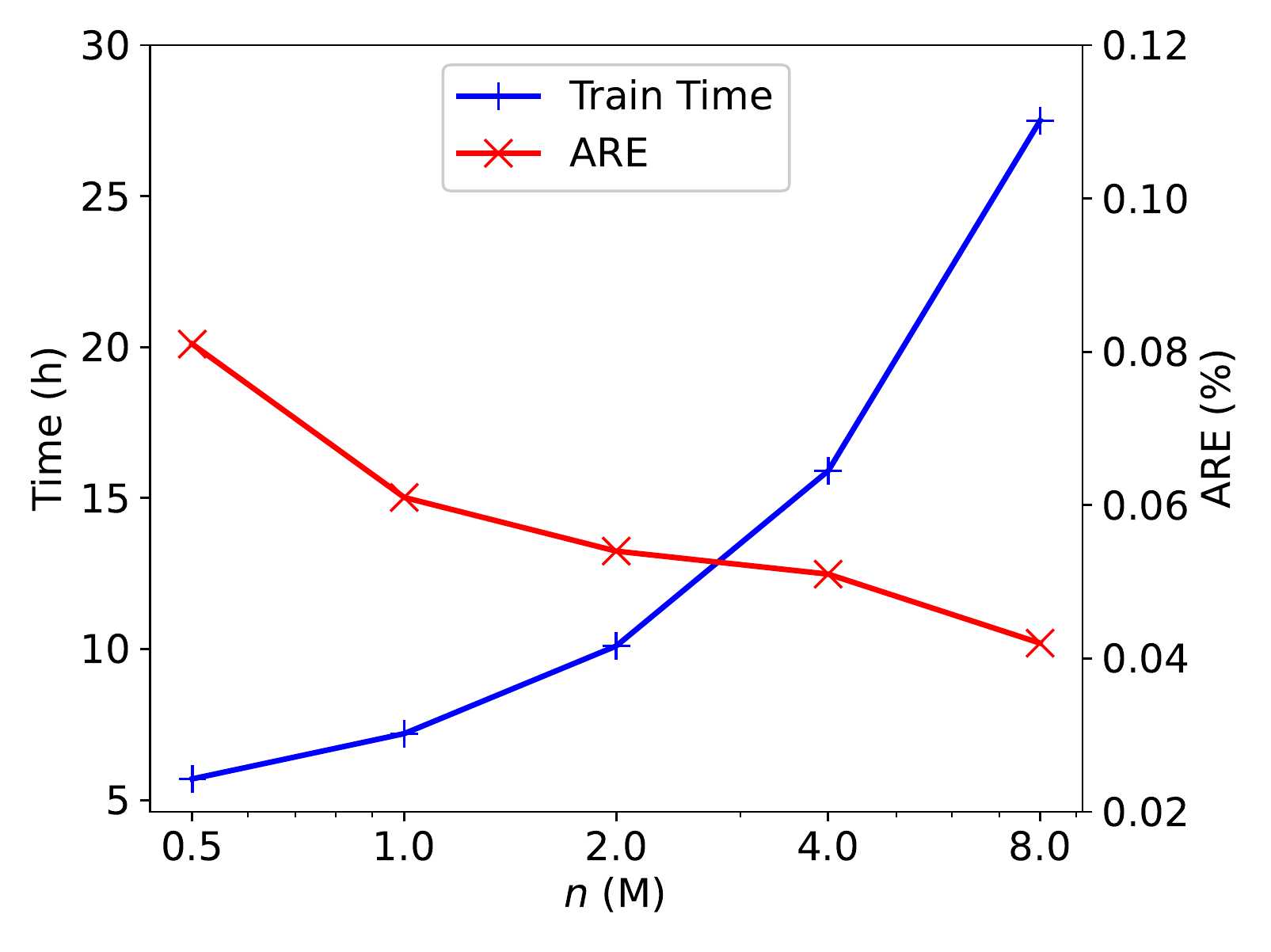}}
	\subfigure[\KXREVIEW{Relationship (cit-PT)}]{
	    \label{fig:relation-light}
		\includegraphics[width=0.235\textwidth]{sup_figs/relation-light.pdf}}
    \vspace{-2mm}
    \caption{Results of counting triangles under light deletion scenario: (a) the ARE under different orderings of the stream; (b) the ARE under different $M$'s; (c) the training time and ARE under different size of training graphs; (d) the relationship between the weights and the number of triangles. }
    \vspace{-4mm}
    \label{fig:res-light}
\end{figure*}

\begin{table}[th]
\vspace{-4mm}
    \centering
    {\KXREVIEW{
    \caption{Training time (hours) of counting triangles ($\triangle$) and wedges ($\wedge$) on four real datasets under light deletion scenario. }
    \vspace{-2mm}
    \label{tab:train-light}
    \begin{tabular}{c|c|c|c|c|c|c|c|c}
    \toprule
    \textbf{Graph} & \multicolumn{2}{c|}{cit-HE} & \multicolumn{2}{c|}{com-DB} & \multicolumn{2}{c|}{soc-TX} & \multicolumn{2}{c}{web-SF} \\
    \midrule
    Pattern $H$ & $\triangle$ & $\wedge$ & $\triangle$ & $\wedge$ & $\triangle$ & $\wedge$ &$\triangle$ & $\wedge$\\
    \midrule
    Time (h)  & 15.6 & 15.1 & 7.2 & 6.4 & 8.5 & 7.2 & 11.7 & 10.9 \\
    \bottomrule
    \end{tabular}
    }}
    \vspace{-4mm}
\end{table}

\smallskip\noindent
\textbf{(5) Training.}
We report the training time of counting triangles and wedges on four datasets under light deletion scenarios in Table~\ref{tab:train-light}. The results of how the size of training graphs affects the model performance under light deletion scenario are shown in Figure~\ref{fig:train-light}.

\smallskip\noindent
\textbf{(6) Relationship between an edge's weight in \texttt{WSD-L} and its associated subgraph counts. }
Figure~\ref{fig:relation-light} shows the relationship results of counting triangles under light deletion scenario.

\begin{table}[ht]
\vspace{-4mm}
\scriptsize
    \centering
    \caption{Results on transferability of \texttt{WSD-L} under light scenarios. }
    \vspace{-2mm}
    \label{tab:transfer-light}
    \begin{tabular}{c|ccccc|c}
    \toprule
        (\textbf{Training}) & cit-HE & com-DB & soc-TX & web-SF & synthetic & \texttt{WSD-H}\\
    \midrule
        cit-PT  & \textbf{0.171} & 0.213 & 0.192 & \underline{0.188} & 0.204 & 0.221\\
        com-YT & \underline{0.055} & \textbf{0.051} & 0.059 & 0.056 & 0.058 & 0.059\\
        soc-TW & 0.681 & 0.702 & \textbf{0.576} & \underline{0.631} & 0.732 & 0.762\\
        web-GL & \underline{0.063} & 0.068 & 0.065 & \textbf{0.061} & 0.067 & 0.069\\
    \bottomrule
    \end{tabular}
    \vspace{-4mm}
\end{table}

\smallskip\noindent
\textbf{(7) Transferability of \texttt{WSD-L}.}
Table~\ref{tab:transfer-light} shows the ARE results on transferability test of counting triangles under light deletion scenario.



\begin{table}[th]
\vspace{-4mm}
    \centering
    \caption{ARE (\%) on ablation study of \texttt{WSD-L} of counting triangles on four real datasets. }
    \vspace{-2mm}
    \label{tab:ablation}
    \begin{tabular}{C{1.5cm}C{1.8cm}C{1.8cm}C{1.8cm}}
    \toprule
    (\textbf{Massive}) & \texttt{WSD-L} (Max)  & \texttt{WSD-L} (Avg) & \texttt{WSD-H}\\
    \midrule
    cit-PT  & \textbf{0.075} & \underline{0.081} & 0.083 \\
    com-YT  & \textbf{0.048} & \underline{0.050} & 0.053 \\
    soc-TW & \textbf{0.400} & \underline{0.540} & 0.710\\
    web-GL  & \textbf{0.031} & \underline{0.033} & 0.037  \\
    \midrule
    \midrule
    (\textbf{Light}) & \texttt{WSD-L} (Max)  & \texttt{WSD-L} (Avg) & \texttt{WSD-H} \\
    \midrule
    cit-PT  & \textbf{0.171} & \underline{0.189} & 0.221 \\
    com-YT  & \textbf{0.051} & \underline{0.052} & 0.059  \\
    soc-TW & \textbf{0.564} & \underline{0.649} & 0.762 \\
    web-GL  & \textbf{0.063} & \underline{0.067} & 0.069  \\
    \bottomrule
    \end{tabular}
    \vspace{-4mm}
\end{table}

{\KXREVIEW{
\smallskip\noindent
\textbf{(8) Ablation study.} 
%
We conduct ablation study on the definitions of $v_j$ and $s_k^v$ (Eqs.~(\ref{eq:vj}) and (\ref{eq:sv})) as follows. 
{\KXREVIEW{We change the $\max$ function in Eq.~(\ref{eq:vj}) to the average function, i.e., $v_j = Avg\{i_j \mid e_{i_j} \in J, J\in \mathcal{H}_k \}$, and concatenate these values together
to form $s_k^v$.}}
We denote the algorithm using the definition in Eq.~(\ref{eq:vj}) by \texttt{WSD-L} (Max) {and the} algorithm using the above definition by \texttt{WSD-L} (Avg). 
The results of the comparison are shown in Table~\ref{tab:ablation}. 
\texttt{WSD-L} (Max) is more effective than \texttt{WSD-L} (Avg). 
One possible reason is that \texttt{WSD-L} (Max) can extract the temporal information of the recent edges directly. 
%
}}

\smallskip\noindent

\begin{figure}[th]
\vspace{-4mm}
    \centering
    \subfigure[ARE (massive deletion scenario)]{
	    \label{fig:beta-m}
		\includegraphics[width=0.23\textwidth]{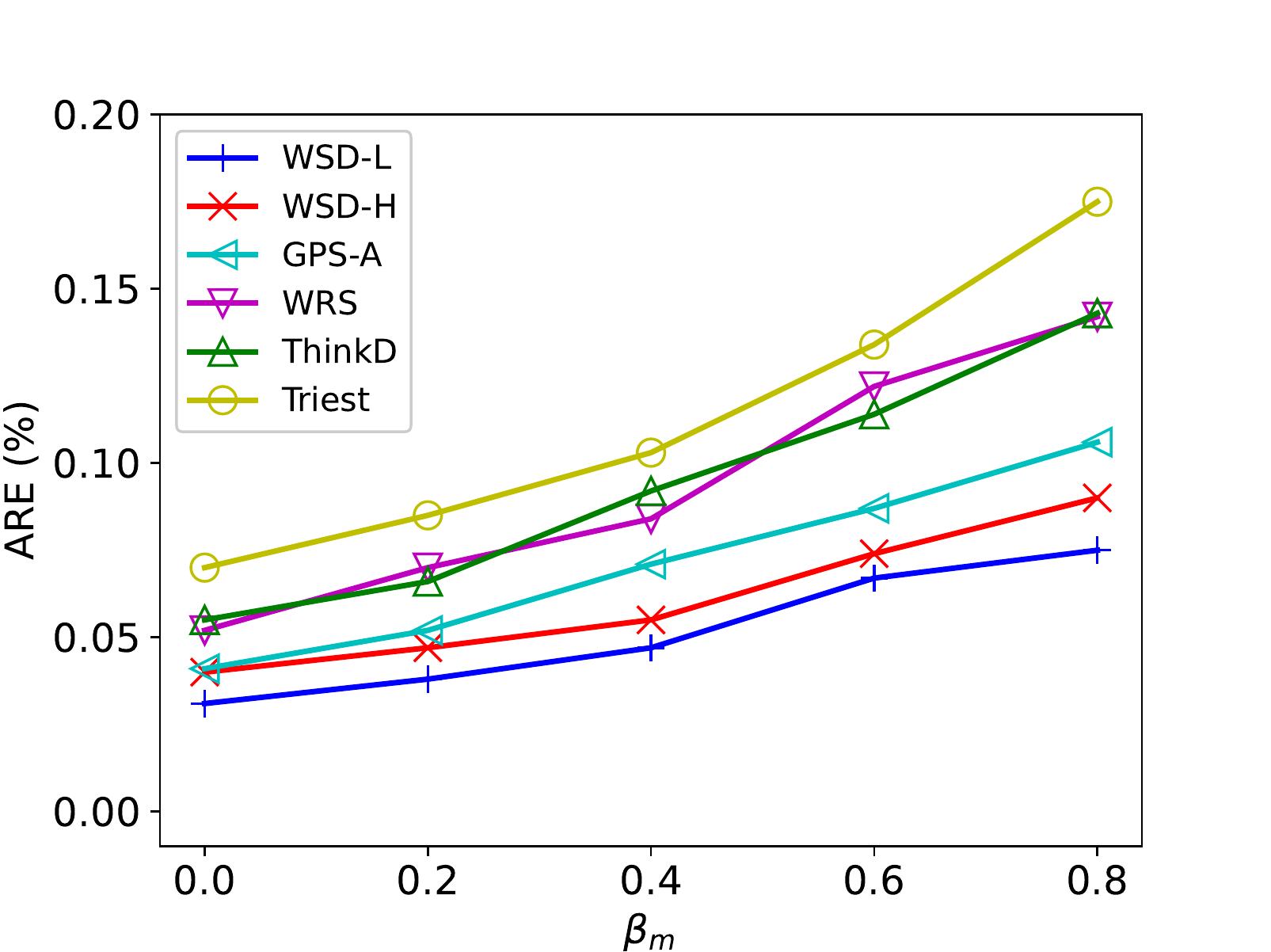}}
	\subfigure[ARE (light deletion scenario)]{
	    \label{fig:beta-l}
		\includegraphics[width=0.23\textwidth]{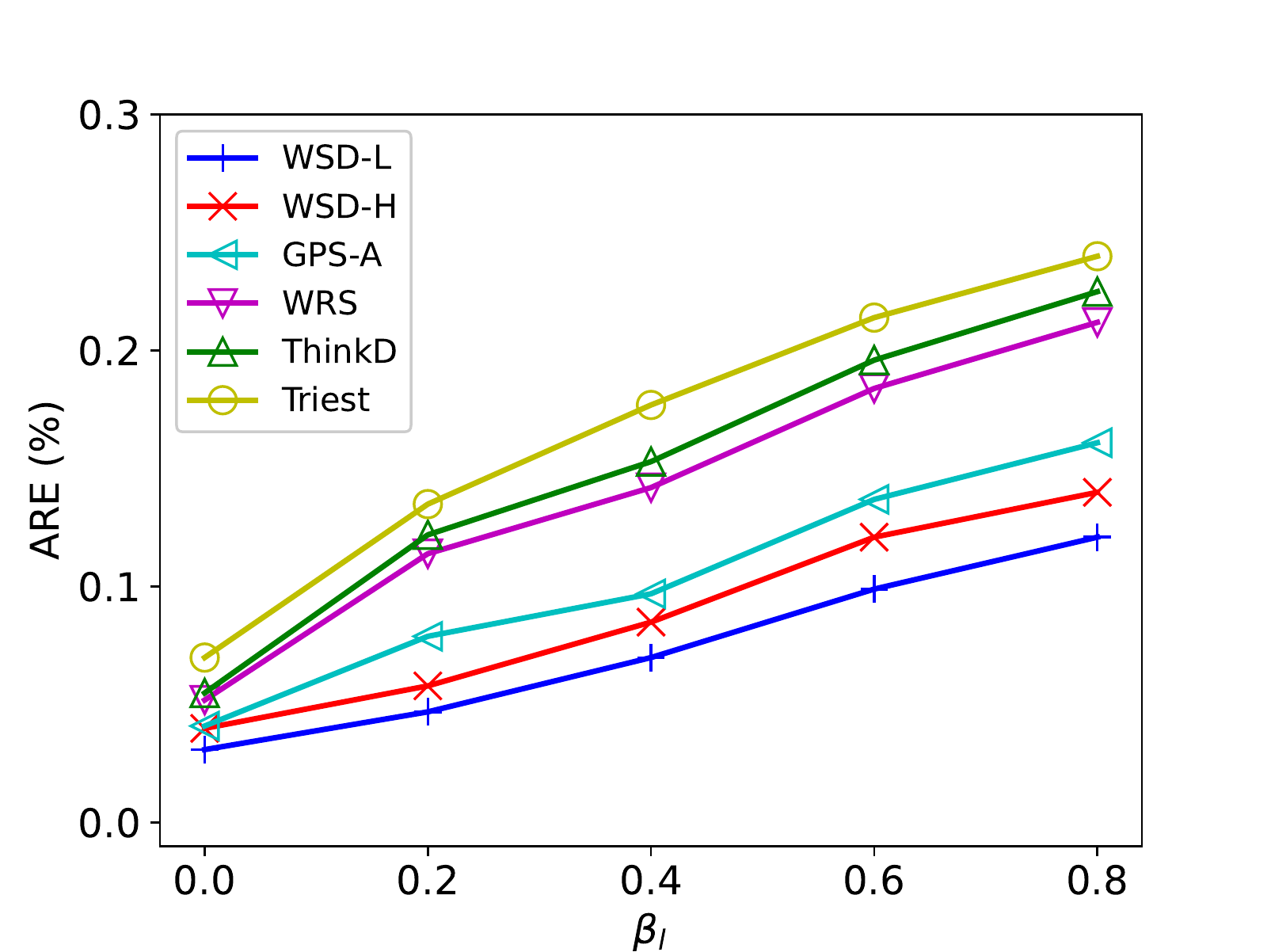}}
    \vspace{-2mm}
    \caption{Results on varying $\beta_m$ and $\beta_l$ on cit-PT, reporting the ARE. }
    \vspace{-4mm}
    \label{fig:beta}
\end{figure}

\smallskip\noindent
\textbf{(9) Effects of $\beta_m$ and $\beta_l$. } We study the effects of $\beta_m$ {\CHENG (resp. $\beta_l$), which indicates the probability of a deletion event or the potion of deletion events,} by varying its value from the range $\{0, 0.2, 0.4, 0.6, \bm{0.8}\}$ {\CHENG (resp. $\{0, \bm{0.2}, 0.4, 0.6, 0.8\}$)}. 
%
{\KAIXIN{For each parameter, we retrain the policy following the steps that are introduced in Section~\ref{subsec:setp-up}. }}
Figure~\ref{fig:beta} shows the results of counting triangles on cit-PT. 
We observe that the ARE increases as {\CHENG $\beta_m$ and $\beta_l$} increase. 
{\CHENG There can be two possible reasons.}
First, as the number of deletion {\CHENG events} increases, the size of the stream becomes larger. Thus, the estimation results on a stream with more edge events would be less accurate. Second, there are {\CHENG fewer} triangles at the end of the stream when there are more deletions. 
{\CHENG Therefore, the ARE (which is a relative error measurement) tends to be larger even the absolute error is not changed.}
%
Under different scenarios and with different values of $\beta_m$ and $\beta_l$, \texttt{WSD-L} and \texttt{WSD-H} outperform other algorithms due to their weight-sensitive nature.

\end{document}